\documentclass[12pt]{amsart}

\usepackage{ifpdf}
\usepackage[dvips]{graphicx} 
\usepackage{subfig}
\usepackage{color}
\usepackage[usenames,dvipsnames]{xcolor}
\usepackage{hyperref}
\usepackage{times}
\usepackage{enumitem} 
\usepackage{todonotes}

\usepackage{mathtools}  
\usepackage{amsfonts} 
\usepackage{amssymb}
\usepackage{amsthm} 
\numberwithin{equation}{section}

\usepackage[T1]{fontenc} 
\usepackage[utf8]{inputenc}  

\usepackage{tikz}
\tikzstyle{every picture}+=[remember picture]
\usetikzlibrary{calc}

\newtheorem{thm}{Theorem}[section]
\newtheorem{corollary}[thm]{Corollary}
\newtheorem{lemma}[thm]{Lemma}

\newcommand{\f}{\mathcal{F}}

\newcommand{\bra}{\langle}
\newcommand{\ket}{\rangle}
\newcommand{\kb}{ \ket \bra }

\newcommand{\indP}{x}
\newcommand{\indQ}{y}
\newcommand{\indR}{z}

\DeclareMathOperator{\tr}{Tr} 
\DeclareMathOperator{\Cat}{Cat}

\newcommand{\N}{\mathbb{N}}

\newcommand{\R}{\mathbb{R}}
\newcommand{\C}{\mathbb{C}}
\newcommand{\E}{\mathbb{E}}

\newcommand{\F}{\mathcal{F}}

\newcommand{\SC}{\operatorname{SC}}

\DeclareMathOperator{\Tr}{Tr}
\DeclareMathOperator{\trace}{Tr}

\DeclareMathOperator{\id}{id}

\DeclareMathOperator*{\Ex}{\mathbb{E}}
\DeclareMathOperator{\NC}{NC}
\DeclareMathOperator{\dist}{dist}
 
\begin{document}



\title[Partial transpose of random quantum states]{Partial transpose of random quantum states: exact formulas and meanders}  

\author{Motohisa Fukuda}
\address{Zentrum Mathematik, M5, 
Technische Universität München, \linebreak Boltzmannstrasse 3, \mbox{85748 Garching,} Germany}
\email{m.fukuda@tum.de}

\author{Piotr \'Sniady}
\address{Zentrum Mathematik, M5,
Technische Universität München, \linebreak
Boltzmannstrasse 3,
85748 Garching, Germany \newline \indent
Institute of Mathematics, Polish Academy of Sciences, \linebreak
\mbox{ul.~\'Sniadec\-kich 8,} 00-956 Warszawa, Poland
 \newline
\indent 
Institute of Mathematics,
University of Wroclaw,  \linebreak \mbox{pl.\ Grunwaldzki~2/4,} 50-384
Wroclaw, Poland
} 
\email{Piotr.Sniady@tum.de, Piotr.Sniady@math.uni.wroc.pl}

\begin{abstract}
We investigate the asymptotic behavior of 
the empirical eigenvalues distribution of
the partial transpose of a random quantum state. 
The limiting distribution was previously investigated 
via Wishart random matrices indirectly (by approximating the matrix of trace $1$ by the Wishart matrix of random trace)
and shown to be the semicircular distribution 
or the free difference of two free Poisson distributions,
depending on how dimensions of the concerned spaces grow. 
Our use of Wishart matrices gives exact combinatorial formulas for the moments of the partial transpose of the random state.
We find three natural asymptotic regimes in terms of geodesics on the permutation groups. 
Two of them correspond to the above two cases; the third one turns out to be a new matrix model for the meander polynomials. 
Moreover, we prove the convergence to the semicircular distribution together with its extreme eigenvalues
under weaker assumptions, and 
show large deviation bound for the latter.  
\end{abstract}

\maketitle

\section{Introduction}     
In this paper, we investigate asymptotic behavior of 
the empirical eigenvalues distribution of the
\emph{partial transpose} of the random \emph{quantum state} (positive semidefinite Hermitian matrix of trace one). 
This problem originated from the field of \emph{quantum information theory} in relation to detecting \emph{entanglement} in a bipartite system. 
Non-entangled states, called \emph{separable} states,  are necessarily positive semidefinite after partial transpose \cite{Peres1996},
where the latter property is called \emph{Positive Partial Transpose}, abbreviated as \emph{PPT}. 
The converse statement is not true except for bipartite states on $\C^2 \otimes \C^2$ and $\C^2 \otimes \C^3$ \cite{Horodecki1996}. 
Here, the partial transpose is made by
writing the bipartite matrix as a Kronecker product and transposing each block. 
Therefore, the generic eigenvalue distribution of the partial transpose of a random quantum state
is interesting and especially the behavior of the minimal eigenvalue is important.  

Mathematically, we investigate the following problem. 
Take three complex vector spaces $\C^l$, $\C^m$ and $\C^n$ with $l,m,n \in \N$ and
define $\rho:= \trace_{\C^l} |v\rangle\langle v|$ for 
a uniformly random unit vector $|v\rangle$ in the product space $\C^l \otimes \C^m \otimes\C^n$.
Space $\C^l$ is called the \emph{environment}, space $\C^m \otimes \C^n$ is called the \emph{system}, 
and individual spaces $\C^m$ and $\C^n$ correspond to individual parts of the bipartite system. 
We will present the details of this construction in Section \ref{subsec:our-model}.
This way of inducing measure on mixed quantum states was investigated in
\cite{ZyczkowskiSommers2001, BengtssonZyczkowski2006}.  
Then, we study the asymptotic behavior of the eigenvalues of its partial transpose $\rho^{\Gamma}$
where the transpose acts only on the space $\C^n$.

Although quantum states correspond to positive semidefinite Hermitian matrices of trace one,
Aubrun in \cite{Aubrun2012} used the normal Wishart matrix model
to approximate a random quantum state. The trace of such a Wishart matrix is a random variable which converges to one only asymptotically.
Aubrun showed that 
the empirical eigenvalues distribution converges to the semicircular distribution 
as the dimension of the spaces grow in such a way that
$l \propto mn$ and $m \propto n$. 
Later, Banica and Nechita \cite{BanicaNechita2012,BanicaNechita2012a} showed with the same model that
the limiting distribution is the free difference of two free Poisson distributions
 in the regime where
the dimension $m$ of one of the parts of the system is fixed and the dimension $l$ of environment and the dimension $n$ of the other system grow proportionally.  

By contrast, we look into this problem more directly by considering the matrix $\rho:=\frac{1}{\Tr W} W$, where $W$ is a Wishart matrix.
Since the trace $\Tr W$ is a $\chi^2$-random variable, independent from $\rho$, the problem of calculating moments of (the partial transpose of) $\rho$ is reduced to calculating appropriate moments of the Wishart matrix $W$. This idea was sketched briefly already by Aubrun \cite[Section 8.2]{Aubrun2012}. In this way we obtain exact combinatorial formulas involving summation over the symmetric group.
We find three natural types of geodesics in the Cayley graph of the symmetric group which
yield interesting asymptotic distributions. 
Two of them correspond to the above mentioned cases (\cite{Aubrun2012} and \cite{BanicaNechita2012,BanicaNechita2012a}),
and the remaining new case turns out to be related to the \emph{meander polynomials} \cite{DiFrancescoGolinelliGuitter}. 

Our paper is organized as follows. 
After explaining necessary mathematical techniques (in particular, free probability) and
our precise mathematical model in Section \ref{preliminary}, 
we analyze the regime, where $l,m,n$ grow such that $l\propto mn$ in Section \ref{non-singular}.
In \cite{Aubrun2012} one requires $m \propto n$ but we drop this condition to get the same limiting measure in Section \ref{limit-dist},
although we need some weak conditions to show the convergence of  extreme eigenvalues and their large deviation 
in Sections \ref{convergence-extreme} and \ref{convergence-speed}. 
Then, it is shown in Section \ref{fixed-env} that our random matrix model yields the meander polynomials.
The connection to free Poisson distribution is presented in Section \ref{fixed-transpose}. 
Section \ref{conclusion} contains the concluding remarks.

\section{Preliminaries}\label{preliminary} 

\subsection{Free probability, noncrossing partitions and permutations}

\begin{figure}[tbp]
\centering
\ifpdf
\includegraphics{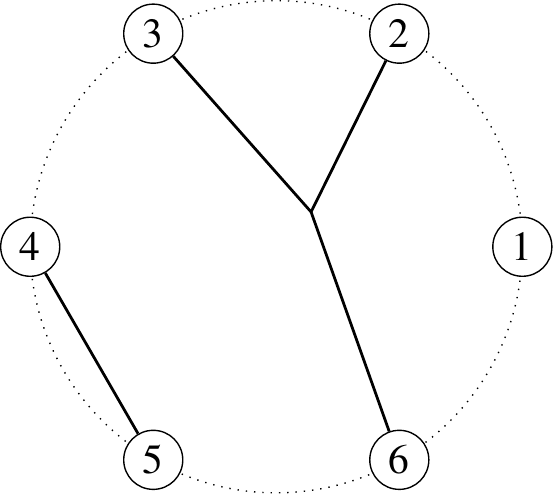}
\else
\includegraphics{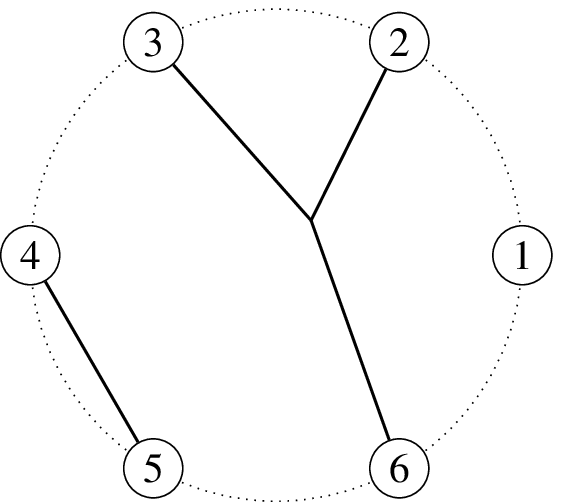}
\fi
\caption{A graphical representation of a noncrossing partition $\tau=\big\{ \{1\}, \{2,3,6\}, \{4,5\} \big\}$.}
\label{fig:noncrossing}
\end{figure}

In this paper we will use some basic notions from free probability theory. A good treatment of this topic is given in the book \cite{NicaSpeicher}.
We will recall briefly the most important notions.

\subsubsection{Noncrossing partitions}
A \emph{noncrossing partition} $\tau$ is a partition of the set $[p]:=\{1,\dots,p\}$ with a property that if $a < b < c <d \in [p]$ are such that $a, c$ belong to the same block of $\tau$ and $b, d$ belong to the same block of $\tau$ then $a,b,c,d$ belong all to the same block of $\tau$. Noncrossing partitions can be represented graphically as noncrossing connections between points arranged on a circle, see Figure \ref{fig:noncrossing}. 

The set of noncrossing partitions of $[p]$ will be denoted by $\NC(p)$.
We also use the notation
\begin{equation*}
\NC_{i_1, \ldots , i_l}(p) :=  \big\{\tau \in \NC(p) : 
|c| \in \{i_1, \ldots , i_l\} \quad \forall c \in \tau\big\}
\end{equation*}
for the set of noncrossing partitions of $[p]$ with a restriction on sizes of the blocks.

\subsubsection{Permutations}
We denote the permutation group of $p$ elements by $S_p$. 
For a permutation $\alpha \in S_p$ we define
$\# \alpha$  to be the number of cycles in $\alpha$ and define its \emph{length}
$|\alpha|$ as the minimum number of factors necessary to write $\alpha$ as a product of  transpositions (we are allowed to use \emph{arbitrary} transpositions, not only Coxeter generators).  The Cayley graph of the symmetric group (where as generators we take \emph{all} transpositions, not only Coxeter generators) defines the distance
\[\dist(\alpha,\beta):=|\alpha^{-1}\beta|\]  on $S_p$
and  
\[\# \alpha = p -|\alpha| \] holds for all $\alpha \in S_p$. 

For given permutations $\alpha,\gamma\in S_p$ we define \emph{the geodesic $\alpha\rightarrow\gamma$} as 
the set of all permutations which are on the geodesics between $\alpha$ and $\gamma$:
\[
\alpha\rightarrow\gamma := \big\{\beta \in S_p :  \dist(\alpha,\beta)+\dist(\beta,\gamma) = \dist(\alpha, \gamma)\big\}.
\]
If $\beta$ belongs to this geodesic, we denote it by $\alpha\rightarrow\beta\rightarrow\gamma$. 

\subsubsection{Noncrossing partition and permutations}

We will recall now the construction of Biane \cite{Biane1997}.
We consider the canonical full cycle 
\begin{equation}
\label{eq:pi}
 \pi=\pi_p:=(1,2,\dots,p)\in S_p.  
\end{equation}
For a given partition $\tau$ of $[p]$ and $i\in[p]$ we define
$\big( t(\tau) \big) (i)\in[p]$ to be the element in the same block of $\tau$ as $i$ which is \emph{after} $i$ (with respect to the cyclic order), see Figure \ref{nc-counterclockwise}. Formally, $\big( t(\tau) \big) (i)$ is the first element of the sequence $\pi(i), \pi^2(i)=\pi\big( \pi(i) \big), \dots$ which belongs to the same block of $\tau$ as $i$. It is easy to check that $t(\tau)\in S_p$.

\begin{figure}[tbp]
\centering
\subfloat[]{
\label{nc-counterclockwise}
%
\ifpdf
\includegraphics{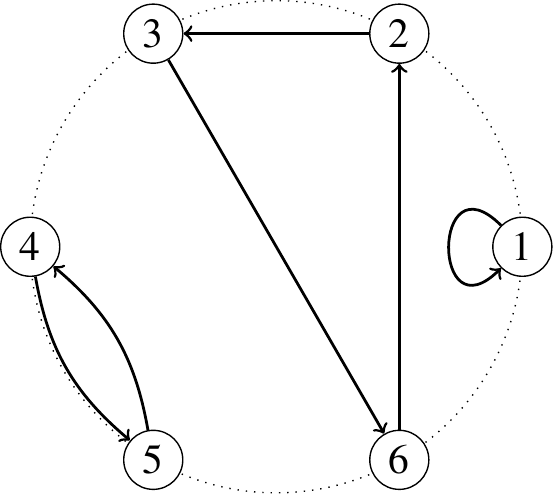}
\else
\includegraphics{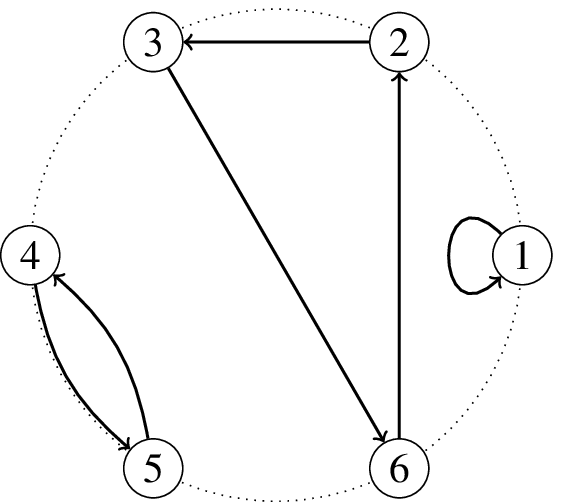}
\fi
}
\hfill
\subfloat[]{
\label{nc-clockwise}
\ifpdf
\includegraphics{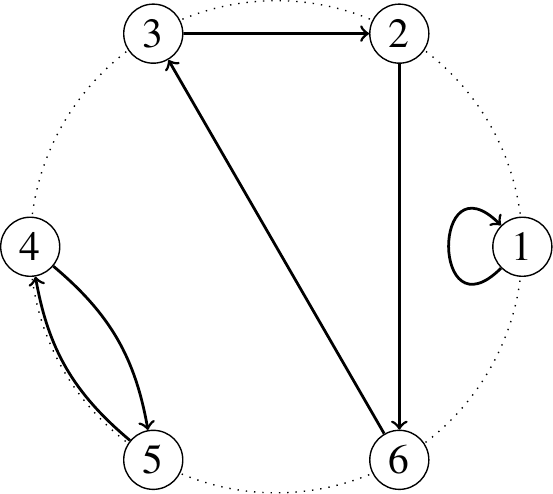}
\else
\includegraphics{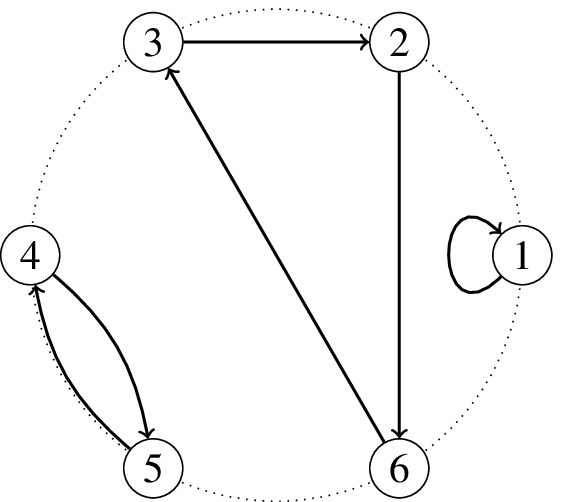}
\fi
}
\caption{\protect\subref{nc-counterclockwise} 
Graphical representation of the permutation $t(\tau)=(1)(236)(45)$ from the geodesic $\id \rightarrow \pi=(12\dots6)$ corresponding to the noncrossing partition $\tau$ from Figure \ref{fig:noncrossing}. \newline
\protect\subref{nc-clockwise}
Graphical representation of the permutation $\big( t(\tau)\big)^{-1}=(1)(632)(54)$ from the geodesic $\id\rightarrow \pi^{-1}=(654\dots1)$ corresponding to the noncrossing partition $\tau$ from Figure \ref{fig:noncrossing}.
}
\label{fig:permutations}
\end{figure}

It is easy to check that 
$\big( t(\tau) \big)^{-1} (i)\in[p]$ is the element in the same block of $\tau$ as $i$ which is \emph{before} $i$ (with respect to the cyclic order), see Figure \ref{nc-clockwise}.

The relationship between noncrossing partitions and geodesics in the Cayley graph of the symmetric group is given by the following result. Comparison of Figure \ref{fig:noncrossing} and Figure \ref{fig:permutations}
 is probably the best way to give an intuitive meaning to this relationship.

\begin{lemma}
\label{lem:Biane}
\
\begin{enumerate}[label=\emph{\alph*})]
\item \label{item:BianeA}
The map $\tau\mapsto t(\tau)$ is a bijective correspondence between $\NC(p)$ and the geodesic $\id\rightarrow \pi$.

\item \label{item:BianeB}
The map $\tau\mapsto \big( t(\tau) \big)^{-1}$ is a bijective correspondence between $\NC(p)$ and the geodesic $\id\rightarrow \pi^{-1}$.

\item \label{item:BianeC}
The map $\tau\mapsto t(\tau)=\big( t(\tau) \big)^{-1}$ is a bijective correspondence between $\NC_{1,2}(p)$ and the intersection of geodesics $\big( \id\rightarrow \pi\big) \cap \big( \id\rightarrow \pi^{-1}\big)$.
\end{enumerate}
\end{lemma}
\begin{proof}
Part \ref{item:BianeA} was proved by Biane \cite{Biane1997}.
Part \ref{item:BianeB} follows in a similar way be replacing permutation $\pi$ by $\pi^{-1}$.

We will prove now part \ref{item:BianeC}. If $\tau\in\NC_{1,2}(p)$ then each cycle of $t(\tau)$ has length $1$ or $2$, thus $t(\tau)=\big( t(\tau) \big)^{-1}$. From part \ref{item:BianeA} and \ref{item:BianeB} it follows that $t(\tau)=\big( t(\tau) \big)^{-1}$ belongs to each of the two geodesics.

In order to show surjectivity, let $\sigma \in \big( \id \rightarrow \pi\big) \cap \big( \id \rightarrow \pi^{-1}\big)$. From part \ref{item:BianeA} we know that there exists $\tau\in\NC(p)$ such that $t(\tau)=\sigma$. From part \ref{item:BianeB} we know that there exists $\tau'\in\NC(p)$ such that $\big( t(\tau') \big)^{-1}=\sigma$.
Since permutations $t(\tau)$ and $t(\tau')$ differ just by orientation of the cycles, it follows that $\tau=\tau'$. Thus $t(\tau) = \sigma = \big( t(\tau) \big)^{-1}$; it follows that $\sigma$ is an involution, therefore each block of $\tau$ consists of $1$ or $2$ elements. Therefore we showed existence of $\tau\in\NC_{1,2}(p)$ such that $t(\tau)=\sigma$. 
\end{proof}

\subsubsection{Genus functions}
\label{subsubsec:genus}
It will be convenient to consider the following two non-negative, integer functions on $S_p$ given by:
\begin{align*}
2g_p^{(1)}(\alpha) &:= \dist(\id ,\alpha)+\dist(\alpha,\pi^{-1})-\dist(\id ,\pi^{-1}),    \\
2g_p^{(2)}(\alpha) &:= \dist(\id ,\alpha)+\dist(\alpha,\pi)-\dist(\id ,\pi),   
\end{align*}
for $\alpha \in S_p$.
They are called \emph{genus functions}; they measure how  
the paths via $\alpha$ are longer than the geodesic distance between
$\id$ and $\pi^{-1}$ or $\pi$.

\subsubsection{Free cumulants}

For a probability measure $\mu$ on the real line $\R$ having all moments finite we consider its \emph{free cumulants} $\big( k_p(\mu) \big)_{p=1,2,\dots}$ given by the following implicit relationship with the moments:
\begin{equation} 
\label{eq:moment-cumulant}
m_p(\mu) := \int x^p \ d\mu(x) = \sum_{\tau\in\NC(p)} \prod_{b\in\tau} 
k_{|b|}(\mu).
\end{equation}
For example,
\begin{align*}
 m_1 = & k_1, \\
 m_2 = & k_2+ k_1^2, \\
 m_3 = & k_3+3 k_1 k_2+k_1^3,\\
 m_4 = & k_4+4 k_1 k_3+2 k_2^2+6 k_1^2 k_2+k_1^4.
\end{align*}
Free cumulants are a fundamental tool of the combinatorial approach to free probability theory \cite{NicaSpeicher}.

\subsubsection{Semicircular distribution}   
The \emph{semicircular distribution} with mean $M$ and standard variation $\sigma$,
which will be denoted by $\SC_{M,\sigma}$, has the following density:
 \[
\frac{d \SC_{M,\sigma}}{dx} = \frac{1}{2\pi\sigma^2}\sqrt{4\sigma^2 - (x-M)^2}
\qquad \text{for } |x-M| \leq 2 \sigma.
\]
This measure has a compact support $[M-2\sigma,M+2\sigma]$. 
The free cumulants of this measure are given by
\[
k_p(\SC_{M,\sigma}) = 
\begin{cases}
M        & \text{if }p=1, \\
\sigma^2 & \text{if }p=2, \\
0        & \text{otherwise.}
\end{cases} 
\]
In the special case when $M=1$, the moment-cumulant formula \eqref{eq:moment-cumulant} implies that the moments of this measure are given by:

\[ m_p(\SC_{1,\sigma}) = \sum_{\tau \in \NC_{1,2}(p)} \sigma^{2\ (\text{$\#$ of blocks of length $2$ in $\tau$})}.
\]
Note that any $\tau\in\NC_{1,2}(p)$ can be identified with an involution $\tau:=t(\tau)=\big( t(\tau) \big)^{-1}\in S_p$.
By this identification, we have
\begin{equation}
\label{moment_shift2}
m_p(\SC_{1,\sigma}) = \sum_{\tau \in \NC_{1,2}(p)} \sigma^{2\ |\tau|}. 
 \end{equation}

\subsubsection{Free Poisson distribution} 
Let $\lambda \geq 0$ and $\alpha \in \R$. 
Then,
the \emph{free Poisson distribution} with rate $\lambda$ and jump-size $\alpha$
is defined to have the following probability density $\nu$ on $\R$: 
\begin{equation}\label{Poisson}
\nu_{\lambda,\alpha} = 
\begin{cases} (1-\lambda) \delta_0 + \tilde\nu & \text{if $0\leq \lambda <1$,} \\ 
\tilde \nu&\text{if $1<\lambda <\infty$.}
\end{cases}
\end{equation}
Here, $\tilde \nu_{\lambda,\alpha}$ is the measure supported on the interval
$\left[\alpha (1-\sqrt{\lambda})^2, \alpha (1+\sqrt{\lambda})^2\right]$
with the density:
\[
\tilde \nu_{\lambda,\alpha} (t) =\frac{1}{2\pi\alpha t} \sqrt{4\lambda\alpha^2 - (t- \alpha(1+\lambda))^2}\ dt. 
\]
Importantly, free cumulants of this distribution are particularly simple: 
$$k_p = \lambda \alpha^p.$$

When $\alpha =1$, the free Poisson distribution is called in particular
\emph{Mar\v{c}enko-Pastur distribution} (with variance $1$ and parameter $1/\lambda$) \cite{MarcenkoPastur1967}.

\subsection{Meanders} 
Suppose we have an infinite straight river with $2p$ bridges. 
Then, a \emph{meander} (of order $p$) is a collection of closed self-avoiding and noncrossing roads
passing through all of the bridges; in other words,
a meander of order $p$ consists of a number of loops crossing a straight line at $2p$ points, see Figure \ref{figure:meander}. 
Another formulation of this object is via noncrossing pair-partitions. 
If one crosses a bridge then next this person must cross another bridge in order to come back to the original side of the river;
these choices can be represented by an element of $\NC_2(2p)$ for each side of the river. 
Therefore, equivalence classes of meanders are represented by the elements of $\NC_2(2p)\times \NC_2(2p)$, i.e.,~pairs of noncrossing pair-partitions.
Meanders have been investigated in relation to folding polymers.
We refer the reader to \cite{DiFrancescoGolinelliGuitter, DiFrancesco01} for details.

\begin{figure}[tbp] 
 \begin{center}
\ifpdf
\includegraphics[width=0.9\textwidth]{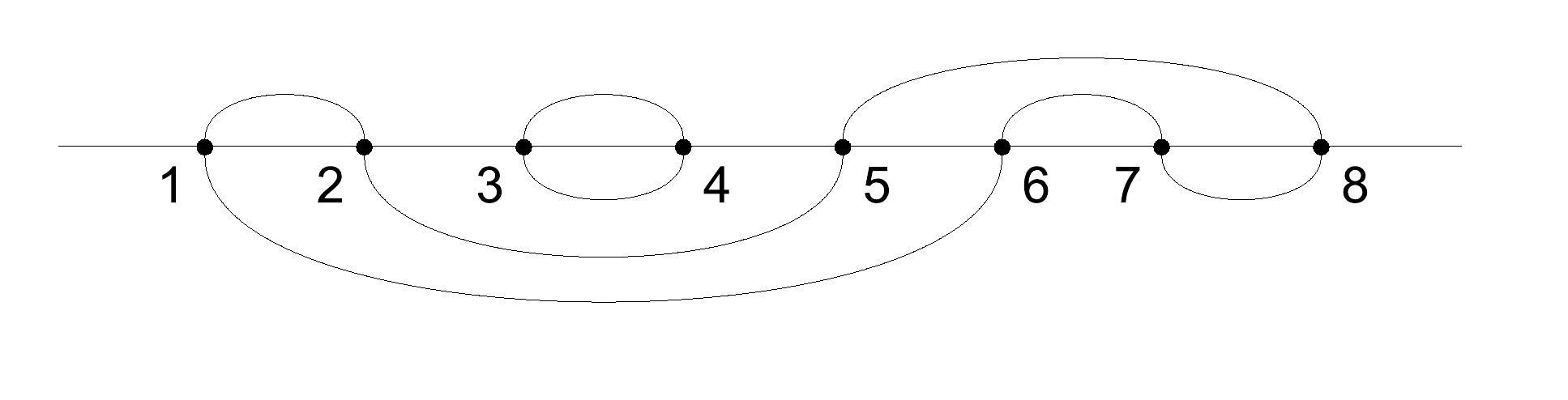}
\else
\includegraphics[width=0.9\textwidth]{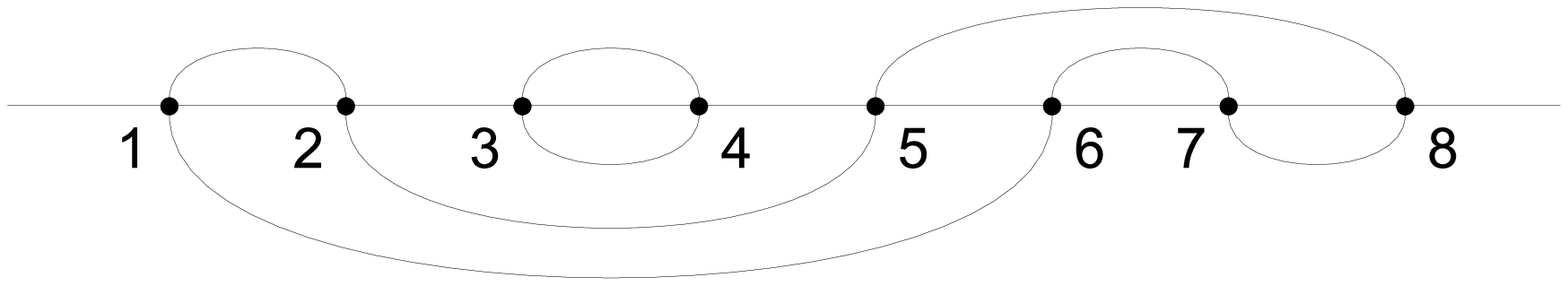}
\fi
\end{center} 
 \caption{Meander of order $p=4$ with $k=2$ connected components corresponding to $\sigma_1 = \big\{ \{1,2\},\{3,4\},\{5,8\},\{6,7\} \big\}$ and 
$\sigma_2 =\big\{ \{1,6\},\{2,5\},\{3,4\},\{7,8\}\big\}$.}  
 \label{figure:meander} 
\end{figure}

For each $k\in [p]$, we define $M_p^{(k)}$ to be the number of nonequivalent meanders with $k$ connected components.  
Then, the \emph{meander polynomial} is defined as
\[
M_p(x) := \sum_{k=1}^{p} x^k M_p^{(k)}.
\] 
Here, each bridge has weight $x$. 
There are some known matrix models for this polynomial, but
we think that ours is the simplest one.
We come back to this problem in Theorem \ref{connection-meanders}.

\subsection{Our model} 
\label{subsec:our-model}

Suppose we have three complex vector spaces $\C^l$, $\C^m$ and $\C^n$ with $l,m,n \in \N$ and 
take the uniformly random unit vector $|v \ket$ in the product space $\C^l \otimes \C^m \otimes\C^n$.
Then, the corresponding random pure state $|v \kb v|$ on $\C^l \otimes \C^m \otimes\C^n$
induces a random mixed state $\rho = VV^*$ on $\C^m \otimes\C^n$ by tracing out the space $\C^l$. 
Here, we use the usual isomorphism $|v \ket \mapsto  V$ between $\C^l \otimes \C^m \otimes\C^n$ and
$\mathcal M_{mn,l} (\C)$.
This model of random mixed quantum states was investigated in \cite{ZyczkowskiSommers2001, BengtssonZyczkowski2006}.  
Then, we study the asymptotic behavior of the eigenvalues of its partial transpose $\rho^{\Gamma}$ with the transpose acting only on the space $\C^n$. 
More precisely, 
let $\{\lambda_i\}_{i=1}^{mn}$ be 
the eigenvalues of the rescaled random matrix $mn\rho^{\Gamma}$;
we define the corresponding empirical eigenvalues distribution
\[
\mu_{mn \rho^\Gamma}:=\frac{1}{mn} \sum_{i=1}^{mn} \delta_{\lambda_i}(x);
\]
our goal is to find the limiting measure in the sense of weak convergence. 
Note that this scaling is used in Section \ref{non-singular} and 
some other scalings are chosen in the following sections.

%

\subsection{Moments of $\rho^\Gamma$}

The following will be our main tool.
\begin{lemma}
\label{lem:key-lemma}
Let $\tau\in S_p$ be a permutation; we denote by $\theta_1,\dots,\theta_\ell$ the lengths of the cycles of $\tau$ (in particular, $\theta_1+\cdots+\theta_\ell=p$). Then
$$ \E  \left( \prod_i \trace \left[\left(mn\rho^{\Gamma}\right)^{\theta_i} \right] \right)= \f(lmn,p) \sum_{\alpha \in S_{p}} l^{-|\alpha|}\ m^{p-|\tau \alpha|}\ n^{p-|\tau^{-1} \alpha|},
$$
where
\begin{equation}
\label{function-f}
\f(D,p) :=   \prod_{i=0}^{p-1} \frac{D}{D+i}
= 1+ O\left( D^{-1} \right).
\end{equation}
\end{lemma}
\begin{proof}
Let $G$ be a $(mn,l)$ matrix with independent, complex $N(0,1)$ entries (i.e.~the real and the imaginary part of each entry are independent, real $N\left(0,\frac{1}{2}\right)$ random variables). The normalized matrix $V:=\frac{1}{\sqrt{ \Tr G G^\star}} G$ corresponds (under the standard isomorphism) to the random unit vector $|v\rangle \in\C^l\otimes \C^m \otimes \C^l$, with the uniform distribution on the unit sphere.
We define the Wishart matrix $W:= G G^\star$; thus as $\rho=V V^\star$ we can take $\rho:=\frac{1}{\Tr W} W$. 

Since $\Tr W$ and $\rho=\frac{1}{\Tr W} W$ are independent random variables (see \cite{Nechita2007}), it follows that
\begin{equation}
\label{eq:towards-chi-sqare}
\E  \left( \prod_i \trace \left[\left(W^{\Gamma}\right)^{\theta_i} \right] \right)
= 
\Ex (\Tr W)^p \cdot  
\E  \left( \prod_i \trace \left[\left(\rho^{\Gamma}\right)^{\theta_i} \right] \right). 
\end{equation}

Since the distribution of $2 \Tr W$ is equal to $\chi^2(2lmn)$, it follows that
\begin{equation} 
\label{eq:chi-square}
\Ex (\Tr W)^p = \frac{1}{2^{p}} \underbrace{(2lmn)(2lmn+2)\cdots(2lmn+2p-2)}_{\text{$p$ factors}}.
\end{equation}

We use the following convention for indexing entries of the matrix $G$: $G=\left( G_{\indP\indQ,\indR} \right)$ with $\indP\in[m]$, $\indQ\in [n]$, $\indR\in [l]$. Thus the entries of $W$ and its partial transpose are given by
\begin{align*} W_{\indP\indQ, \indP'\indQ'} = & \sum_\indR G_{\indP\indQ,\indR}\ \overline{G_{\indP'\indQ',\indR}},\\
(W^\Gamma)_{\indP\indQ, \indP'\indQ'} = &\sum_\indR G_{\indP\indQ',\indR}\ \overline{G_{\indP'\indQ,\indR}}. 
\end{align*}

Thus
\begin{multline} 
\label{eq:bigsum}
 \prod_i \trace \left[\left(W^{\Gamma}\right)^{\theta_i} \right]=\\
\sum_{\indP_1,\dots,\indP_p} 
\sum_{\indQ_1,\dots,\indQ_p}  
\sum_{\indR_1,\dots,\indR_p}
G_{\indP_1 \indQ_{\tau(1)}, \indR_1} \cdots G_{\indP_p \indQ_{\tau(p)}, \indR_p} \cdot \overline{G_{\indP_{\tau(1)} \indQ_{1}, \indR_1}} \cdots \overline{G_{\indP_{\tau(p)} \indQ_p, \indR_p}}.
\end{multline}
We apply Wick formula; since the entries of $G$ are complex Gaussian, the summation is only over pairings which match each non-bared factor with some bared one; each such a pairing can be described by a permutation $\alpha\in S_p$. Therefore
\begin{multline}
\label{eq:Wick} 
\E  \left( \prod_i \trace \left[\left(W^{\Gamma}\right)^{\theta_i} \right] \right)=\\
\shoveleft{ \sum_{\alpha \in S_p}
\left(
\sum_{\indP_1,\dots,\indP_p}  
\prod_{s\in[p]} 
[\indP_s = \indP_{\tau(\alpha(s))}] 
\right) \times} \\
\shoveright{ \left(
\sum_{\indQ_1,\dots,\indQ_p}  
\prod_{s\in[p]} 
[\indQ_{\tau(s)}= \indQ_{\alpha(s)}]
\right) 
\left(
\sum_{\indR_1,\dots,\indR_p}  
\prod_{s\in[p]} 
[\indR_s=\indR_{\alpha(s)}] \right)
=} \\
\sum_{\alpha\in S_p} m^{\# \tau\alpha}\ n^{\#\tau^{-1} \alpha}\   l^{\# \alpha}.
\end{multline}

Combining \eqref{eq:towards-chi-sqare}, \eqref{eq:chi-square} and \eqref{eq:Wick} we get the desired result.
%
\end{proof}

For an integer $p\geq 1$ we consider the random variable
\[
Z_n^{(p)} := \frac{1}{mn} \trace \left[\left(mn\rho^{\Gamma}\right)^p\right] =
\int x^p \ d\mu_{mn\rho^\Gamma}
\]
which is just the appropriate moment of the empirical eigenvalues distribution $\mu_{mn\rho^\Gamma}$.
\begin{corollary}
\label{cor:moments}
For an arbitrary integer $p\geq 1$ the expected value of the corresponding moment of $mn \rho^\Gamma$ is given by 
\begin{align} 
\E Z_n^{(p)}
& = \f(lmn,p) \sum_{\alpha \in S_{p}} l^{-|\alpha|}\ m^{p-1-|\pi \alpha|}\ n^{p-1-|\pi^{-1} \alpha|}
\label{moment}\\
& = \f(lmn,p)
\sum_{\alpha \in S_{p}} \left( \frac{mn}{l}\right)^{|\alpha|}\ m^{-2 g_p^{(1)}}\ n^{-2 g_p^{(2)}},
\label{moment3.1} 
\end{align}
where $\pi$, as usual, is given by \eqref{eq:pi}, $g_p^{(i)}:=g_p^{(i)}(\alpha)$ were defined in Section \ref{subsubsec:genus}.
\end{corollary}
\begin{proof}
It is a direct consequence of Lemma \ref{lem:key-lemma} with $\tau=\pi$ equal to the canonical full cycle.
\end{proof}

Interestingly, the above formula \eqref{moment3.1} gives three regimes with interesting limiting measures. 
The first case, which we investigate in Section \ref{non-singular}, 
is that $l \propto mn$, where three permutations $\alpha,\pi\alpha,\pi^{-1}\alpha$ interact.
The following sections treat 
the other two cases when $l$ or $m$ are fixed
so that just two of the permutations $\alpha,\pi\alpha,\pi^{-1}\alpha$ interact.

\section{The case when dimensions of environment and both system parts are large} 
\label{non-singular}

In this section we investigate the regime 
where $l,m,n \rightarrow \infty$.
Aubrun \cite{Aubrun2012} investigated the case where
$l\propto d^2$ and  $m,n \propto d$ as $d\rightarrow \infty$ via  Wishart random matrices model and 
showed the results which correspond to Theorem \ref{limit-dist} and Theorem \ref{convergence-edge}.
We will show the following results.
Theorem \ref{limit-dist} shows that $mn\rho^\Gamma$ has the limiting distribution 
as long as the dimension $mn$ of the quantum system $\C^m\otimes \C^n$ is proportional to the dimension $l$ of the environment $\C^l$. 
Then, Theorem \ref{convergence-edge} studies the behavior of 
the smallest and largest eigenvalues. 
They turn out to converge to the two corresponding edges of the support of the limiting density
unless $m$ and $n$ grow too differently. 
Finally, large deviation property of the extreme eigenvalues  
is investigated in Theorem  \ref{convergence-speed}.

\subsection{Limiting eigenvalues}\label{convergence-distribution}
We analyze the limit of the empirical eigenvalues distribution of $mn\rho^\Gamma$
when $l \propto mn$ where $l,m,n \rightarrow \infty$. 
In the following, we assume without loss of generality that $m \geq n$.

\begin{thm}\label{limit-dist}
Suppose $\frac{mn}{l} \to a$ with $0\leq a<\infty$ and $m \geq n $. 
Then, as $n \rightarrow \infty$,
the empirical measure of $mn\rho^\Gamma$ converges weakly almost surely to the semicircle distribution $\SC_{1,\sqrt{a}}$. 
Here, we think of $m=m_n$ and $l=l_n$ as sequences which implicitly depend on $n$. 
\end{thm}
Before presenting the proof we remark that 
in the case $a=0$  the limit distribution $\SC_{1,0}=\delta_1$ becomes a delta measure. 
\begin{proof}[Proof of Theorem \ref{limit-dist}.] 
Non-zero contribution to \eqref{moment3.1} in the limit is given only by the summands for which $\alpha$ is on the following two geodesics:
\[ \id \rightarrow \alpha \rightarrow \pi^{-1};\qquad \id \rightarrow \alpha \rightarrow \pi. \]
We apply Lemma \ref{lem:Biane}\ref{item:BianeC}; it follows that
\[
\alpha \in \NC_{1,2} (p).
\]
Hence, we have
\begin{equation}
\label{eq:moments-converge-sc}
\lim_{n \rightarrow \infty}
\Ex Z_n^{(p)}
= \sum_{\alpha \in \NC_{1,2}(p)} a^{|\alpha|} =
m_p(\SC_{1,\sigma})
\end{equation}
with 
\[ 
\sigma^2 = a,
\] 
where we used \eqref{moment_shift2}.
Thus, we proved the convergence in (expected) moments.

To prove almost sure convergence, we will show later that 
\begin{equation}
 \label{as-conv}
\sum_{n=1}^\infty \operatorname{Var} Z_n^{(p)}  
= \sum_{n=1}^\infty {\Ex \left[ \left(Z_n^{(p)} \right)^2 \right]}   
- {\left(\Ex Z_n^{(p)} \right)^2} 
< \infty
\end{equation}
for each $p \in \N$. This result via standard arguments 
(involving Markov inequality and Borel-Cantelli lemma) would imply that
$Z_n^{(p)}$ converges \emph{almost surely} to the appropriate moment of the semicircle distribution. The latter distribution is uniquely determined by its moments, so 
convergence of measures in the sense of moments implies the convergence in the weak sense; this would finish the proof.

It remains to show that \eqref{as-conv} indeed holds true; we shall do it in the following. 
A more careful analysis of \eqref{eq:moments-converge-sc} shows that every term which converges to zero is in fact at most $O\left( n^{-2} \right)$ thus
\[ \Ex Z_n^{(p)}
=  \sum_{\alpha \in \NC_{1,2}(p)} \left( \frac{mn}{l} \right)^{|\alpha|} + O\left(n^{-2}\right).
\]

We consider the permutation
\[ \hat{\pi}:=(1,2,\dots,p)(p+1,p+2,\dots,2p) \in S_{2p}.\]
We apply Lemma \ref{lem:key-lemma} for $\tau=\hat{\pi}$; it follows that
$\Ex \left( Z_n^{(p)} \right)^2$ is equal to the right-hand side of \eqref{moment3.1} in which the summation over $S_p$ is replaced by summation over $S_{2p}$ and permutation $\pi$ is replaced by $\hat{\pi}$; notice that also the definitions of $g_p^{(i)}$ are affected by this replacement.
In an analogous way it follows that 
\[ \Ex \left( Z_n^{(p)} \right)^2
=  \left( \sum_{\alpha \in \NC_{1,2}(p)} \left( \frac{mn}{l} \right)^{|\alpha|} \right)^2 + O\left(n^{-2}\right)
\]
which finishes the proof of \eqref{as-conv}.
\end{proof}

Theorem \ref{limit-dist} shows that 
the limiting empirical distribution has the compact support on the interval
$[1-2\sqrt{a}, 1+2\sqrt{a}]$. 
However, this does not necessarily mean that 
the minimum and maximum eigenvalues converge to the boundaries of this interval. 
We analyze the convergence of extreme eigenvalues in the next section.

\subsection{Behavior of extreme eigenvalues}\label{convergence-extreme}
In this section we analyze the behavior of minimum and maximum eigenvalues. 
Theorem \ref{convergence-edge} below shows that in the regime of Theorem \ref{limit-dist}
the minimum and the maximum eigenvalues of $mn \rho^\Gamma$ actually converge respectively to
$1-2\sqrt{a}$ and $1+2\sqrt{a}$ under an additional condition on growth of $m$ and $n$.

\subsubsection{Convergence of extreme eigenvalues}
\label{conv-ex}

\begin{thm}\label{convergence-edge}
Let assumptions of Theorem \ref{limit-dist} hold true with an additional condition that
$\log m = o(n^{1/6})$.
Then the extreme eigenvalues of $mn\rho^\Gamma$ converge to $1\pm 2\sqrt{a}$ almost surely.
%
%
\end{thm} 
\begin{proof} 
The difficult part of the theorem is to show that
\[ \limsup_{n \rightarrow \infty} \Vert  mn \rho^\Gamma -I \Vert _\infty \leq 2\sqrt{a}\]
holds almost surely. We will do it in the following.

Let $(p_n)$ be the sequence of even numbers such that
$p_n$ is the largest even number such that $2p_n^{12} \max\left\{1,\frac{mn}{l}\right\} \leq n^2$;
we write $p=p_n$ for simplicity. 
Note that then $\log m = o(p)$. 

Hence we have 
\begin{multline*} 
\Ex \Vert  mn \rho^\Gamma -I \Vert _\infty^p \leq \Ex \Vert  mn \rho^\Gamma -I \Vert _p^p
= \\ \Ex \trace \left[( mn \rho^\Gamma -I )^p \right]
\leq 2 mn p^5  \left(2\sqrt{\alpha}+ o(1) \right)^p.  
\end{multline*}
Here, the last inequality follows from Lemma \ref{high-moment} below. 
Therefore, Markov inequality implies that for any $\epsilon>0$
\[
\sum_n \Pr \left\{\Vert  mn \rho^\Gamma -I \Vert _\infty \geq 2\sqrt{a} + \epsilon \right\}
\leq \sum_n 2 mn p^5 \left( \frac{2\sqrt{a} + o(1)}{2\sqrt{a} + \epsilon } \right)^p 
< \infty.
\]
Thus Borel-Cantelli lemma finishes the proof.
\end{proof}

\begin{lemma}\label{high-moment}
If\/ $m \geq n$ and $2p^{12} \max\{1,a\} \leq n^2$, then 
\[
\Ex \trace \left[ \left(mn \rho^\Gamma -I\right)^p \right] \leq
2 mn p^5  \left(2\sqrt{a} + \frac{\sqrt{a} p}{mn} \right)^p,
\] 
where $a:=\frac{mn}{l}$.
\end{lemma} 
Before presenting the formal proof let us make an intuitive remark on this phenomenon. 
Since we now know that $mn\rho^\Gamma$ obeys the shifted semicircle law
in the limit, 
we must more or less have the following:
\[
\frac{1}{mn}\Ex \trace \left[ \left(mn \rho^\Gamma -I\right)^p \right] 
\approx (\sqrt{a} )^p\Cat_{p/2} 
\approx (2\sqrt{a})^p, 
\]
where used the fact that Catalan numbers 
\[ \Cat_{k} = \frac{1}{k+1} \binom{2k}{k} = \int x^{2k}\ d\SC_{0,1}(x) \]
are the moments of the centered semicircular distribution. 
\begin{proof}[Proof of Lemma \ref{high-moment}]
First, we expand 
\[
\frac{1}{mn}\Ex \trace \left[ \left(mn \rho^\Gamma -I\right)^p \right]
= \sum_{k=0}^p \binom{p}{k} (-1)^{p-k} 
\underbrace{\frac{1}{mn}\Ex \trace \left[ \left(mn \rho^\Gamma\right)^k \right] }_{(\spadesuit)}.
\]
Here,  \eqref{moment3.1} implies that
\[
(\spadesuit)
=\f(mnl,k)  \sum_{\alpha \in S_{k}} 
a^{|\alpha|}\ m^{-2g_k^{(1)}} n^{-2g_k^{(2)}}. 
\]

Next, we decompose $S_k$ in the following way:
\[
S_k = \hat S_{k,0} \cup  \cdots \cup \hat S_{k,k},
\]
where
\[
\hat S_{k,t} = \big\{ \alpha \in S_k  \,:\, \# ( \text{non-fixed points of $\alpha$})= t \big\}. 
\]
For example, $S_{k,0} = \{\id\}$ and $S_{k,1} = \emptyset$, etc. 
We also define 
\[
\tilde S_t = \hat S_{t,t} \subset S_t
\]
as the set of permutations without fixpoints.

For a given $\hat\alpha \in \hat S_{k,t}$ we consider its \emph{support}, i.e.~the set of its non-fixed points and the restriction of $\hat\alpha$ to the support. This support has $t$ elements; by a relabeling of these elements, the restriction can be identified with a permutation $\tilde \alpha \in \tilde S_t$. If the relabeling is chosen to be order-preserving, one can easily check (for example, by removing the fixed points one by one) that
\begin{align*}
|\hat \alpha|& =|\tilde\alpha|, \\
\#\left( \hat\alpha \pi_k\right) & = \#\left( \tilde\alpha \pi_t\right),  \\
\#\left( \hat\alpha^{-1} \pi_k\right) & = \#\left( \tilde\alpha^{-1} \pi_t\right),  
\end{align*}
thus the genus functions on $\hat\alpha$ and $\tilde\alpha$ are the same:
\[
g_k^{(i)}(\hat\alpha) = g_t^{(i)}(\tilde\alpha)
\]
for $i=1,2$.  

Hence, 
\begin{multline*}
  \frac{1}{mn}\Ex \left[ \left(mn \rho^\Gamma -I\right)^p \right]
\\
 \shoveleft{ =
\sum_{k=0}^p \binom{p}{k}  (-1)^{p-k} \f(mnl,k)  
\sum_{\alpha \in S_k}a^{|\alpha|} m^{-2g_k^{(1)}(\alpha)} n^{-2g_k^{(2)}(\alpha)}} 
\\
 \shoveleft{ = 
\sum_{k=0}^p \binom{p}{k} (-1)^{p-k} 
 \f(mnl,k) \sum_{t=0}^k \binom{k}{t} 
\sum_{\alpha \in \tilde S_t}a^{|\alpha|}m^{-2g_t^{(1)}(\alpha)}n^{-2g_t^{(2)}(\alpha)} } 
\\
 \shoveleft{ =  
\sum_{t=0}^p \sum_{\alpha \in \tilde S_t}a^{|\alpha|}m^{-2g_t^{(1)}(\alpha)}n^{-2g_t^{(2)}(\alpha)} 
\sum_{k=t}^p   \binom{p}{k}    \binom{k}{t} (-1)^{p-k}  \f(mnl,k) }
\\
 \leq  \sum_{t=0}^p\binom{p}{t} \left( \frac{\sqrt{a}p}{mn} \right)^{p-t} 
\underbrace{ \sum_{\alpha \in \tilde S_t}a^{|\alpha|}m^{-2g_t^{(1)}(\alpha)}n^{-2g_t^{(2)}(\alpha)}}_{(\clubsuit)}.
\end{multline*}
Here, the last inequality comes from Lemma \ref{identity1}.   

Next, define a function $h_t$: 
\begin{equation}
\label{function-ht}
h_t = g_t^{(1)} + g_t^{(2)}.
\end{equation}
For this function we claim that for $\alpha \in \tilde S_t$
\[
\frac{t}{2} \leq |\alpha| \leq \frac{t}{2} + h_t(\alpha). 
\]
Here, the first inequality comes from the definition of $\tilde S_t$ and
the second one is proved similarly as in \eqref{bound-a}.
Then, we have
\begin{eqnarray*}
(\clubsuit)\leq\sum_{\alpha \in \tilde S_t} a^{|\alpha|}n^{-2h_t(\alpha)} 
&\leq& \sum_{h \geq 0} \sum_{\alpha \in T_{t,h}}a^{\frac{t}{2}} \left( \frac{\max\{1,a\}}{n^2}\right)^h \\
&\leq&
\sum_{h \geq 0} 4^{ \frac{t}{2}}\ t^{12h + 5} \ 
a^{\frac{t}{2}} \left( \frac{\max\{1,a\}}{n^2}\right)^h \\
&\leq& p^5 \left( 2\sqrt{a} \right)^t
\sum_{h \geq 0} \left( \frac{p^{12}\max\{1,a\}}{n^2}\right)^h \\
&\leq& 2p^5\left( 2\sqrt{a} \right)^t .
\end{eqnarray*} 
Here, the first inequality comes from $m \geq n$, and 
Lemma \ref{cardinarity} gives the definition of $ T_{t,h}$ and
the third bound.

Thus we showed that 
\[  \frac{1}{mn}\Ex \left[ \left(mn \rho^\Gamma -I\right)^p \right] \leq 
\sum_{t=0}^p\binom{p}{t} \left( \frac{\sqrt{a}p}{mn} \right)^{p-t} 
 2p^5\left( 2\sqrt{a} \right)^t.
\]

Finally, the binomial formula gives our desired result.  
\end{proof}

In the following section,  
we specify how fast those extreme eigenvalues can converge as is in Theorem \ref{convergence-speed}.  
\subsubsection{Speed of convergence}
Below, we show a large deviation bound for the extreme eigenvalues
by using an enhanced version of Levy's lemma introduced in \cite{AubrunSzarekWerner2011}.
We prove the result only for the minimum eigenvalue since
the proof for the maximum eigenvalue is almost the same.

\begin{thm} \label{convergence-speed} 
In the regime of Theorem \ref{convergence-edge}, 
we have the following large deviation bound for the convergence of the
minimum eigenvalue of $mn\rho^\Gamma$, denoted by $\lambda_{\min}$ below, to $1- 2\sqrt{a}$:
there exist three constants $0< \epsilon_1,c_1<1$ and $l_1 \in \N$ such that
\[
\Pr \left\{ \lambda_{\min}  \leq 1- 2\sqrt{a} - \epsilon \right\} \leq \exp \{-c_1 \epsilon^2  l\}
\]
for all $0<\epsilon<\epsilon_1$ and $l \geq l_1$. 
\end{thm} 
\begin{proof}
We denote the unit sphere in $\C^d$ by $\mathcal S_d$.
Then define a function $\lambda_{\min}:\mathcal S_{lmn} \rightarrow \R$ by:
\[
\lambda_{\min}(V) = \left(\text{the minimum eigenvalue of $mn(VV^*)^\Gamma$} \right).
\]
Here, again we use the identification $|v\ket = V$ between the two spaces $\C^l \otimes \C^m \otimes \C^n \simeq \mathcal M_{mn,l}(\C)$.
For our purposes, this function must be modified similarly as in \cite{AubrunSzarekWerner2011}.  
Let $K \subseteq \mathcal S_{lmn}$ be the compact subset of the whole domain of $\lambda_{\min}$ 
defined by
\[
K = \left\{V \in \mathcal S_{lmn} \,:\, \Vert V\Vert _\infty 
\leq \sqrt{\frac{C}{ l} } 
\right\}.
\]
Here, the constant $C>0$ is chosen as in Corollary\ref{cor-HT},
which implies that the probability of its complement is exponentially small in $l$:
\begin{equation} 
\label{complement}
\Pr_{V} \left\{V \in K^C\right \} \leq \exp \{-c \ l \}.  
\end{equation}
holds for large enough $l\in\N$. 
Here, the constant $c>0$ is again as in Corollary\ref{cor-HT}.

Next, take $U,V \in K$ and
let $|u\ket$ be a normalized eigenvector for the smallest eigenvalue of $\rho_U^\Gamma$, where
we denote $\rho_U = UU^*$ and $\rho_V = VV^*$. 
Then, assuming without loss of generality that $\lambda_{\min}(V) \geq \lambda_{\min}(U)$, we have
\begin{eqnarray*}
\lambda_{\min}(V) - \lambda_{\min}(U)
&\leq&  \bra u |mn \rho_V^\Gamma |u \ket - \bra u | mn \rho_U^\Gamma |u \ket \\
&\leq& mn \left\Vert  \rho_V^\Gamma -  \rho_U^\Gamma  \right\Vert_\infty \\
&\leq &  mn\left\Vert  \rho_V^\Gamma -  \rho_U^\Gamma \right\Vert _2    \\
&=&  mn\Vert  \rho_V-  \rho_U  \Vert _2 \\
&\leq& mn (\Vert V\Vert _\infty + \Vert  U \Vert _\infty) \Vert  V-U\Vert _2 \\
&\leq&2a  \sqrt{Cl\,}  \Vert  V-U\Vert _2. 
\end{eqnarray*}
Hence, the Lipschitz constant of $\lambda_{\min}$ on $K$ is
upper-bounded by $ 2a \sqrt{C l\,} $ and we set $C_1 =2a\sqrt{C}$. 

Finally we extend this restricted function to the whole domain by:
\[
\tilde \lambda_{\min} (V) = \inf_{W \in K} \left\{\lambda_{\min} (W) +  C_1 \sqrt{l} \Vert V-W\Vert _2 \right\}
\qquad \text{ for $V \in K^C$}.
\]
This is a modified function of $\lambda_{\min}$ and different from the original 
only on the small domain $K^C$. This implies, via Theorem \ref{convergence-edge}, 
that for each $\epsilon >0$ there exists some $l_0\in\N$ such that 
\[
\left(1- 2\sqrt{a}\right)  - {\rm Median} \left[\tilde\lambda_{\min}\right]  < \frac{\epsilon}{2}
\qquad \text{for all $l \geq l_0$}.
\]

On the other hand, by applying Lemma \ref{Levy-lemma} (Levy's lemma) to $\mathcal S_{lmn}$,
there exists $\tilde c_0 >0$ such that 
\begin{equation}
\label{Levy's bound}
\Pr_V \left\{  {\rm Median} \left[\tilde\lambda_{\min}\right] -\tilde\lambda_{\min}(V) \geq  \frac{\epsilon}{2}  \right\} 
 <  \exp \left\{ -  \tilde c_0 \epsilon^2 l \right\} 
\end{equation}

To finish the proof, 
we apply the union bound method to \eqref{complement} and \eqref{Levy's bound}. 
Set $\epsilon_1>0$ so that the latter dominates the former
and choose appropriate constants $c_1>0$ and $l_1\in \N$ to get the desired result. 
\end{proof}

\subsubsection{Implications for quantum information theory and some remarks}
As already shown by Aubrun \cite{Aubrun2012}, $a=\frac{1}{4}$ is the critical value. 
When $\frac{mn}{l}\rightarrow a < \infty$
the smaller edge of the support of the limit density $\SC_{1,1/2}$ is $1-2\sqrt{a}$,
which is strictly negative if and only if $a>\frac{1}{4}$.
Hence, our random matrix $\rho^\Gamma$ is  generically not PPT when $a>\frac{1}{4}$ and
PPT when $a<\frac{1}{4}$.
Also, probability of having non-generic states is exponentially small in $l$ 
by Theorem \ref{convergence-speed}.  

We remark that Theorem 8.2 in \cite{Aubrun2012} gives large deviation results on PPT property, 
via convex geometry arguments,
for the case $m=n$; they consider only whether the minimum eigenvalue is positive or negative. 
However, our large deviation bound of Theorem \ref{convergence-speed} works 
for the largest and smallest eigenvalues: $1\pm 2\sqrt{a}$ with $a$ any positive number, and 
for any ratio between $m$ and $n$ as long as $\log m = o(n^{1/6})$, 
as is stated in Theorem \ref{convergence-edge} ($m$ and $n$ are interchangeable).  
With our method, this condition on the ratio is really needed
so that the probability of ``bad'' event converges to $0$
at the end of proof of Theorem \ref{convergence-edge}.

\section{The case when the environment space is fixed and its connection to Meanders}
\label{fixed-env} 
In this section, we investigate our model when the dimension $l$ of the environment space $\C^l$ is fixed. 
Unlike the previous regime, 
we do not have double-geodesics any more. 
However, interestingly, this regime
gives a simple matrix model for the meander polynomials,
which is the main result in this section.
A special case of the other result of this section (Theorem \ref{theo:trivial-environment}) has been already proved by Aubrun \cite[Proposition 9.1]{Aubrun2012}, and one should be able to prove Theorem \ref{theo:trivial-environment} itself
based on the same method, but we give another proof for completeness.

\subsection{Our model} 

In this section we investigate the case where 
\[
l= l_0, \quad \frac{m}{n}\to c, 
\]
for fixed $l_0\in\N,c>0$, in the limit as $n\to\infty$. As usual, $m=m_n$ depends implicitly on $n$.  
Then, we are interested in the following empirical distribution of $lm\rho^\Gamma$:
\begin{equation}\label{emp_l}
\mu_{lm\rho^\Gamma}:=\frac{1}{mn} \sum_{i=1}^{mn} \lambda_i,
\end{equation} 
where $\lambda_i$ are the eigenvalues of $lm\rho^\Gamma$. 
Here, we use a different scaling from the one in Section \ref{non-singular}
because the rank of $\rho^\Gamma$ is $l\cdot\min\{m,n\}$.
The moments of \eqref{emp_l} can be written as 
\begin{multline*}
\frac{1}{mn} \Ex_{U \in \mathcal U (lmn)}  \tr
\left[ \left(lm \rho^{\Gamma} \right)^p\right] 
= \\ \left(1+O\left(n^{-2}\right)\right) \sum_{\alpha \in S_p}l^{p-|\alpha|} c^{p-1-|\pi\alpha|} n^{p-2-|\pi\alpha|-|\pi^{-1}\alpha|}. 
\end{multline*}
by using \eqref{moment} and
\begin{eqnarray*}
\text{(the power of $n$)} &=& p-2- (|\pi \alpha| +|\pi^{-1} \alpha| )\\
&\leq& p-2 - |\pi^2|  
= 
\begin{cases}
-1 & \text{if $p$ is odd,}\\
0 & \text {if $p$ is even.}
\end{cases} 
\end{eqnarray*}
Here, as before $\pi = (1,2,\ldots,p)$ is the canonical full cycle. 
This means in particular that all the odd moments vanish. 
When $p$ is even, the bound is satiated if and only if $\alpha$ is on the following geodesic:
\[
\pi^{-1} \rightarrow \alpha \rightarrow \pi.
\]
This implies that, for even $p \in \N$,
\begin{equation} \label{moment_l}
\lim_{n \rightarrow \infty}\frac{1}{mn} \Ex_{U \in \mathcal U (lcn^2)}  \left[ \left(lm \rho^{\Gamma} \right)^p\right] 
=
\sum_{\pi^{-1}\rightarrow \alpha \rightarrow \pi } 
l^{\#\alpha} c^{\#(\pi \alpha)-1}. 
\end{equation} 

Unfortunately, the above general formula seems too complicated at the moment. 
So, we investigate two restricted cases where 
$l =1$ or $c=1$. 
The former treats random pure states on the bipartite system, and,
interestingly, the latter shows a connection to the meander polynomials.

\subsection{Random pure states}
We start with the case $l=1$ of trivial environment space $\C^l$; this corresponds to $\rho$ being a random pure state.
\begin{thm}
\label{theo:trivial-environment}
Let $l=1$. Then the empirical eigenvalues distribution \eqref{emp_l} converges  almost surely as $n\rightarrow \infty$ 
to 
\begin{equation}
\label{eq:answer}
 \left(1-\frac{1}{c}\right) \delta_0 + \frac{1}{c} \mu_{B\sqrt{X_1 X_2}}  
\end{equation}
in the weak topology of probability measures. Here, $\mu_{B\sqrt{X_1 X_2}}$ is the distribution of ${B\sqrt{X_1 X_2}}$, where the distribution of random variables $X_1$ and $X_2$ is the free Poisson distribution $\nu_{1,c}$ as in \eqref{Poisson},
and $B$ takes the value $1$ or $-1$ with probability $1/2$,
and they are all (classically) independent.
\end{thm} 
Before showing the proof let us make some remarks. 
First, when $c<1$ the coefficient of $\delta_0$ from the first summand is negative, 
but this negativity is canceled by the mass from the second summand. 
Secondly, the case when $l=c=1$ was considered by Aubrun \cite[Proposition 9.1]{Aubrun2012} which reads that the limiting distribution is the product of two independent random variables, each with the
semicircular distribution $\SC_{0,1}$ (with mean $0$ and variance $1$).
Since if $Y$ is a random variable with a semicircle law $\SC_{0,1}$ then its square $Y^2$ has the Poisson distribution $\nu_{1,1}$, this is a special case of Theorem \ref{theo:trivial-environment}.
However, one should be able to recover Theorem \ref{theo:trivial-environment} via this method as well. 
Also, analysis of Schmidt coefficients of $\psi$, where $\rho=|\psi\rangle\langle\psi|$ 
 explains why the operator $\rho^\Gamma$ has at most $\min(m,n)^2$ non-zero eigenvalues for $l=1$
and why the distribution \eqref{eq:answer} has the atom at zero. 

\begin{proof}[Proof of Theorem \ref{theo:trivial-environment}]
For a technical reason, we consider the following rescaled empirical distribution:
\begin{equation}\label{rescaled-emp}
\frac{c}{mn} \sum_{i=1}^{mn} \lambda_i. 
\end{equation} 
Let $p=2q$; we consider permutation $\alpha\in S_p$ which contributes to \eqref{moment_l} so that
permutation $\tau := \pi \alpha$ belongs to the following geodesic: 
\[
\id \rightarrow \tau \rightarrow \pi^2 = (1,3,\ldots,p-1)(2,4,\ldots,p).
\]
Then the limit moment of \eqref{rescaled-emp} is equal to $c \times  \eqref{moment_l}$, which can be calculated as follows: 
\[
c \times \eqref{moment_l}=\sum_{\id \rightarrow\tau\rightarrow\pi^2} c^{\#\tau}  
=\left( \sum_{\tilde \tau \in \NC(q)} c^{\#\tilde \tau} \right)^2.
\]
Moreover, 
\[ 
 \sum_{\tilde \tau \in \NC(q)} c^{\#\tilde \tau} 
= \sum_{\tilde \tau \in \NC(q)} \prod_{V \in \tau} c 
\]
coincides with the appropriate moment of 
the free Poisson distribution $\nu_{c,1}$ (with rate $c$ and jump size $1$). 
Hence, for the even moments, square root of the product of two classically independent free Poisson random variables 
gives the right moments.
However, since all the odd moments vanish, 
we must add the factor $B$ to recover our desired moments
to get the limit distribution of \eqref{rescaled-emp}. 
After rescaling back this distribution,
the additional atom $(1-\frac{1}{c}) \delta_0$ does not change the moments $m_p$ of the measure   for $p=1,2,\dots$ but takes care of the correct value of the moment $m_0$ (the total mass of the measure). This shows convergence in expected moments.

Almost sure convergence can be proven similarly as in the proof of Theorem \ref{limit-dist}. 
Since the limit is compactly supported, the converges in moments implies convergence in the weak sense.
\end{proof}

\subsection{Meander polynomials with our model}
Next, we consider the case where $c=1$,
where our model gives the meander polynomials. 

\begin{thm}\label{connection-meanders}
If $\frac{m}{n}\to c=1$ and $q \in \N$ then the $2q$-th moment of $\mu_{lm\rho^\Gamma}$  converges as $n\rightarrow \infty$ 
to the meander polynomial $M_q (l)$. 
\end{thm} 
\begin{proof} 
First, since $c=1$
\begin{equation}\label{moment-meander} 
\eqref{moment_l} = \sum_{\tau_1,\tau_2} l^{\# [\pi^{-1} (\tau_1 \oplus \tau_2) ]}.
\end{equation}
Here, $\tau_1 \in \NC(\{1,3,\ldots,p-1\})$ and $\tau_2 \in \NC(\{2,4,\ldots,p\})$.

Next, we recall the well-known bijection between $\NC(q)$ and $\NC_2(2q)$. 
We represent a noncrossing partition as in Figure \ref{figure:NC}. 
We add two points $i_{-}$ and $i_{+}$ for both sides of each $i \in [q]$,
left and right respectively. 
Then we connect $i_{+}$ and $j_{-}$ if $\alpha(i) =j$. 
The example is drawn in Figure \ref{figure:NC2}, where we also use arrows to show the action of the permutation $\alpha$. This procedure is commonly known as \emph{fattening}.

\begin{figure}[tbp]  
  \centering
  \subfloat[]{\label{figure:NC}
\ifpdf
\includegraphics[width=0.45\textwidth]{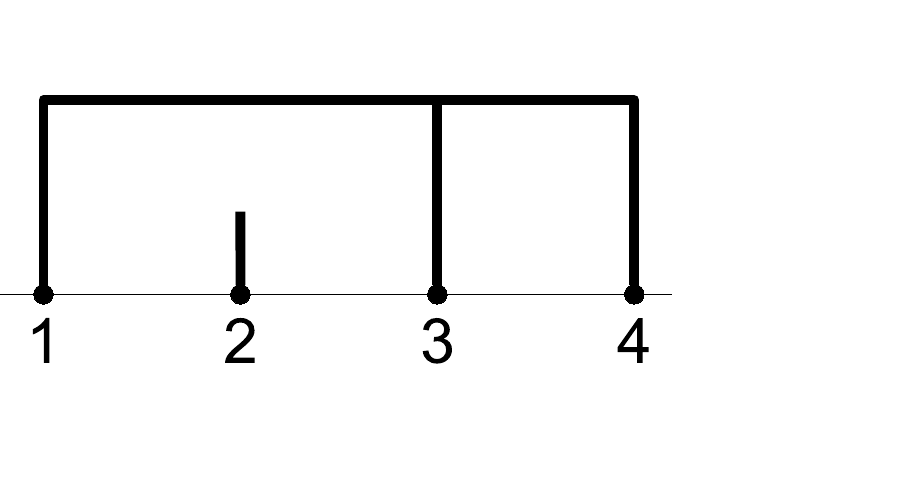}
\else
\includegraphics[width=0.45\textwidth]{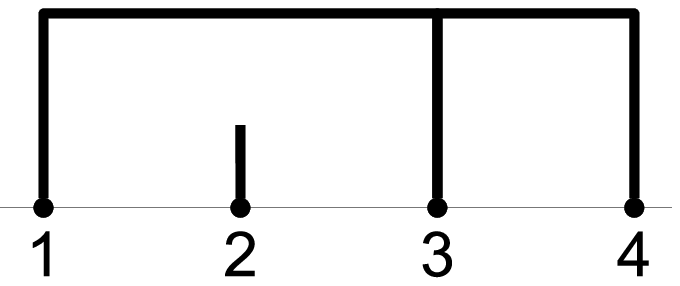}
\fi}
\hfill     
  \subfloat[]{\label{figure:NC2}
\ifpdf
\includegraphics[width=0.45\textwidth]{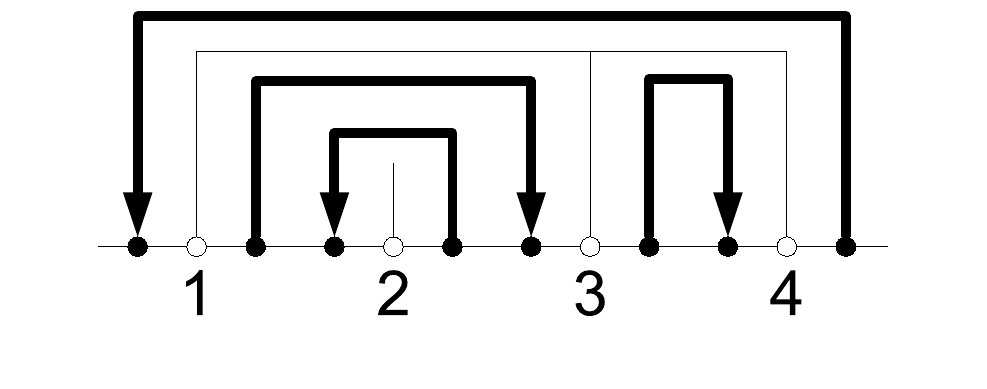}
\else
\includegraphics[width=0.45\textwidth]{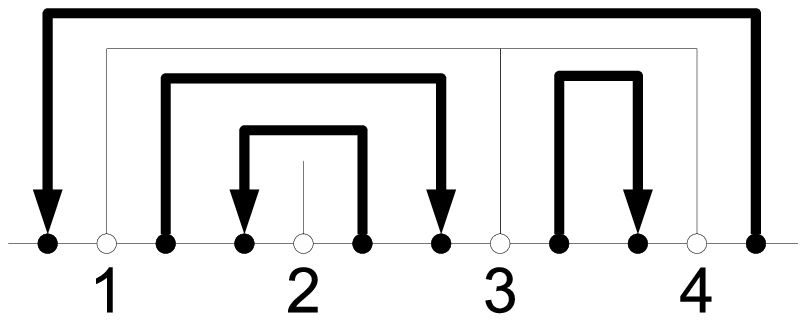}
\fi}
  \caption{\protect\subref{figure:NC} graphical representation of a noncrossing partition $ \big\{\{1,3,4\}\{2\}\big\}\in\NC(4)$, \protect\subref{figure:NC2} the corresponding noncrossing pair-partition from $\NC_2(8)$.}  
  \label{}    
\end{figure}

\begin{figure}[tbp] 
 \begin{center}
\ifpdf
\includegraphics[width=0.9\textwidth]{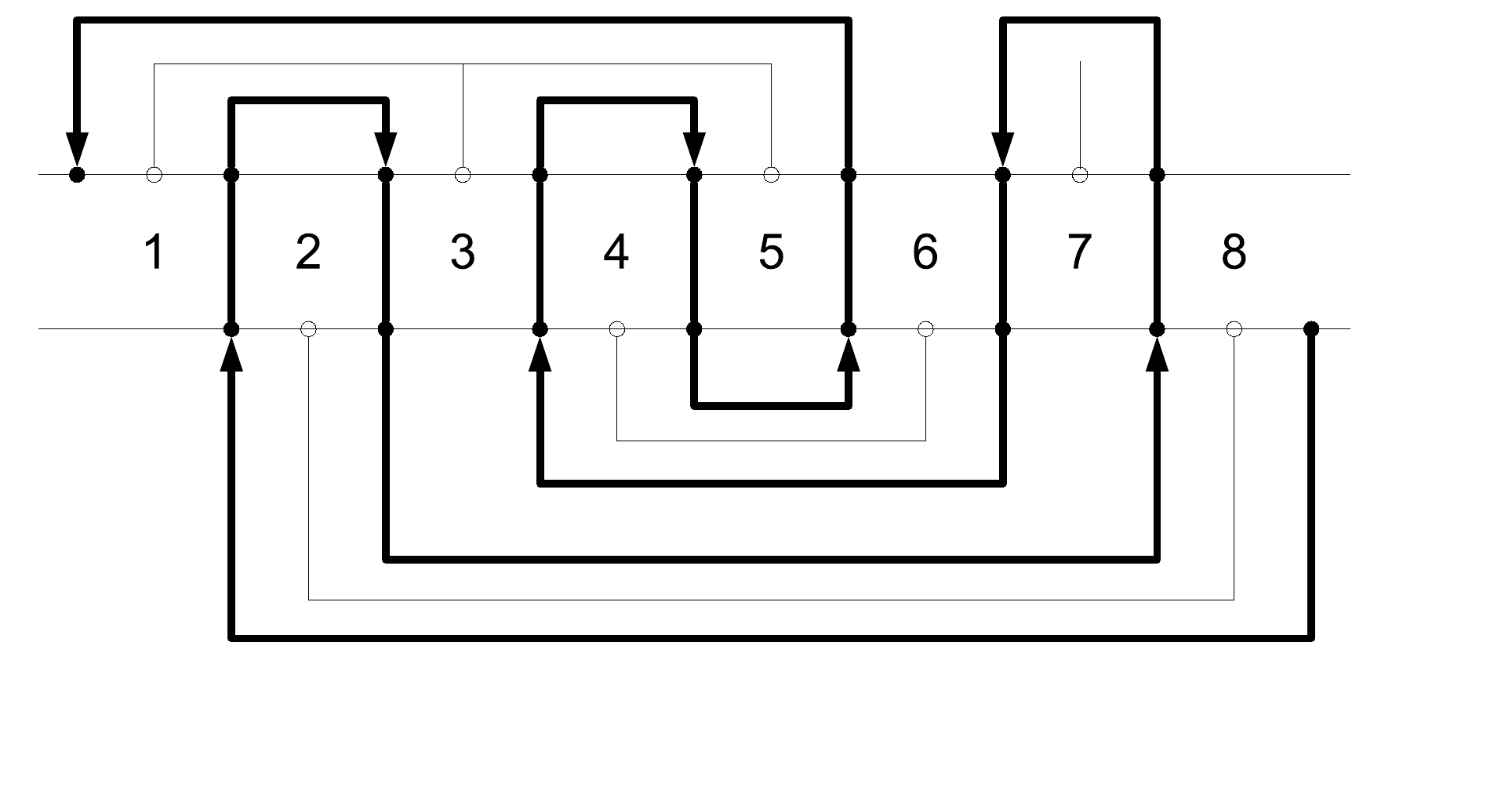}
\else
\includegraphics[width=0.9\textwidth]{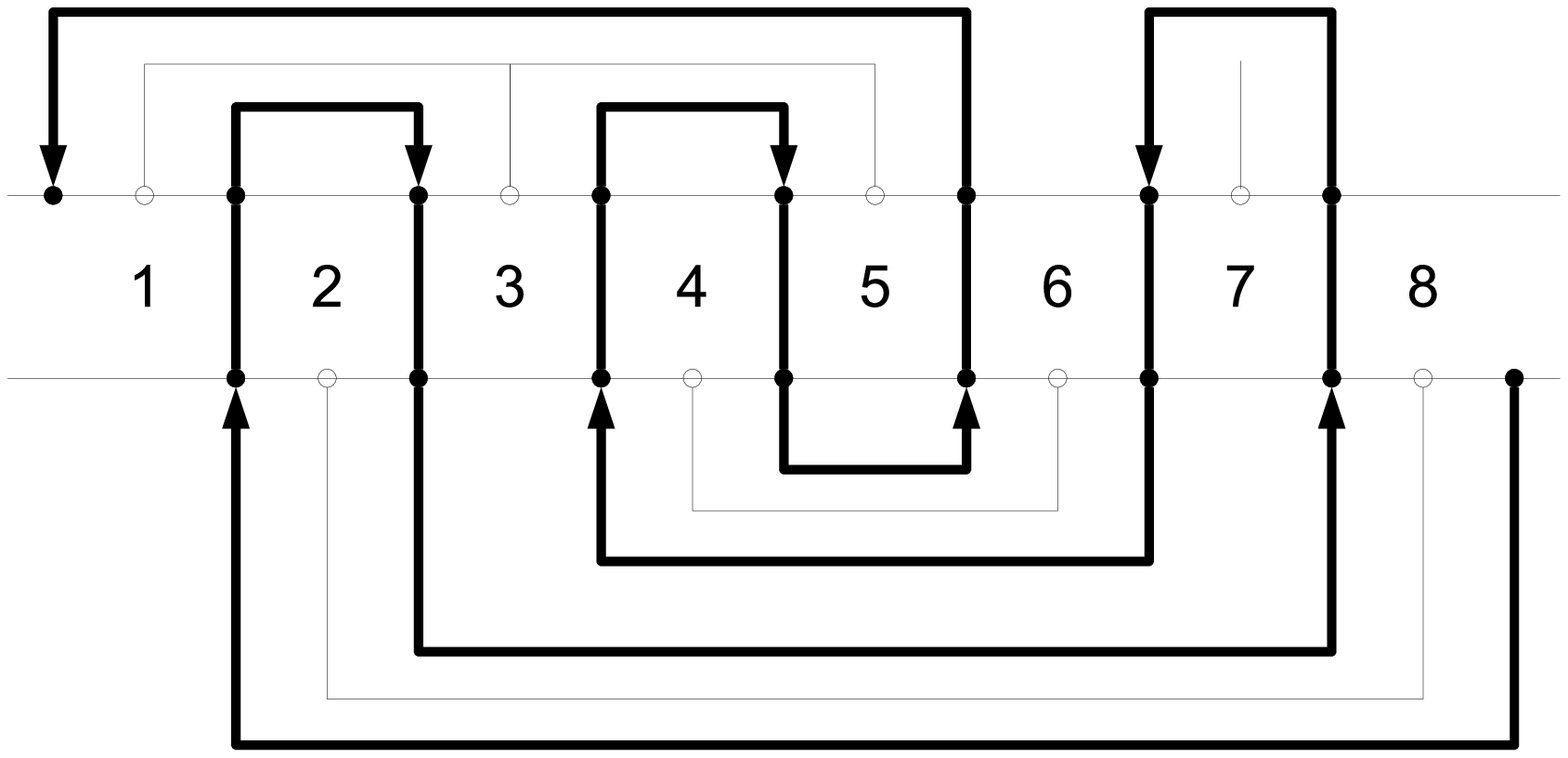}
\fi
\end{center} 
 \caption{Graphical representation of noncrossing partitions $\tau_{1} = (1,3,5)(7)$ and $\tau_2 =(2,8)(4,6)$}   
 \label{figure:meander-NC2} 
\end{figure}

Finally, we will calculate  $\# [\pi^{-1} (\tau_1 \oplus \tau_2)] $. 
To understand the loop-structure of $\pi^{-1} (\tau_1 \oplus \tau_2)$ we note that
$\tau_1$ and $\tau_2$ act in turn; 
$\tau_1$ acts on odd numbers and $\tau_2$ even numbers, but $\pi^{-1}$ switches parities.
For this reason, we suggest the following graphical representation, see Figure \ref{figure:meander-NC2}. 
\begin{enumerate}
\item We draw two parallel horizontal straight lines with odd-numbered points on the upper line and 
 even-numbered points  on the lower line.  
\item Draw the graphical representation for $\tau_1$ above the upper line and
the graphical representation for $\tau_2$ below the lower line. 
\item Identify respectively $(2i+1)_{-}$ and $(2i)_{+}$, and then  $(2i)_{-}$ and $(2i-1)_{+}$.  
\end{enumerate} 
Here, note that the last step corresponds to the action of $\pi^{-1}$. 
An example with  is drawn in Figure \ref{figure:meander-NC2}.
Since we can think of $\tau_1$ and $\tau_2$ as elements of $\NC_2(2q)$, 
identifying $1_{-1}$ and $2q_{+}$ reduces \eqref{moment-meander} into the meander polynomials $M_q(l)$. 
\end{proof}  

To compare our result with an existing matrix model for the meander polynomial,
we quote a result of Di Francesco
\cite[Eqs.~(6)--(9)]{DiFrancesco01}:
\[
M_q(l) = \lim_{N \rightarrow \infty} \frac{1}{N^2} 
\sum_{a_1,a_2,\ldots,a_{2q}=1}^{l} 
\E \left[\left |\trace \left[B_{a_1}B_{a_2} \cdots B_{a_{2q}}\right] \right|^2  \right].
\]
Here,  $\{B_1, \ldots,B_l\}$ are i.i.d. ~random $N \times N$ Hermitian matrices such that
\[
\E \left[(B)_{xy}(B)_{zw} \right] = \frac{1}{N} \delta_{xw}\delta_{yz}. 
\]
On the other hand, our model of partial transposed random quantum states can be modified 
for the meander polynomial in the following way:
\begin{corollary}
Take random complex Gaussian matrices $G$ of $n^2 \times l$, 
where $G_{xy} = a_{xy}+ b_{xy} \mathrm{i}$ and
$a_{xy}, b_{xy}$ are i.i.d. ~real Gaussian distribution with mean $0$ and variance $\frac{1}{2n}$. 
Then, 
\[
 M_q(l) =    \lim_{n \rightarrow \infty} \frac{1}{n^2} \trace \E \left[\left(\left(GG^*\right)^\Gamma\right)^{2q}\right].
\]
\end{corollary}

\section{The case when one of the system parts is fixed} \label{fixed-transpose}
In this section, we review the case where
one of the spaces $\C^m$ or $\C^n$ is fixed. 
Without loss of generality we may assume that $\C^m$ is fixed. 
This case was investigated by Banica and Nechita in \cite{BanicaNechita2012,BanicaNechita2012a}  
via approximation by the complex Wishart matrix and
they proved that the limiting measure is the free difference of two free Poisson distributions. We will present a new proof of this result. 

Suppose 
\[
\frac{l}{n}\to b, \quad m=m_0, 
\] 
where $b>0$ and $m_0 \in \N$ are fixed constants. 
Then, the restatement of \cite{BanicaNechita2012,BanicaNechita2012a}  in our setting is:
\begin{thm}[Banica, Nechita]
The empirical distribution of $m l\rho^\Gamma$ converges weakly, as $n\rightarrow\infty$, almost surely to 
the probability measure of the free difference of free Poisson distributions
with parameters $b(m_0 \pm 1)/2$.
\end{thm} 

\begin{proof}
First, formula \eqref{moment} gives 
\begin{multline*}
\frac{1}{mn}\Ex_{U \in \mathcal{U}(lmn)} \trace [(ml\rho ^{\Gamma})^p]
= \\
\left(1+O\left(n^{-2}\right)\right)\sum_{\alpha \in S_{p}} b^p (bn)^{-|\alpha|} m^{p-1-|\pi \alpha|} n^{p-1-|\pi^{-1} \alpha|}. 
\end{multline*}
Then, the power of $n$ is bounded as:
\[
p-1-(|\alpha|+|\alpha^{-1}\pi|) \leq p-1- |\pi| =0.
\]
This implies the following geodesic formula for the limit:
\begin{eqnarray}
\lim_{n \rightarrow \infty}\frac{1}{mn} \Ex_{U \in \mathcal{U}(lmn)} \trace [(ml\rho ^{\Gamma})^p]
&=& \sum_{\id \rightarrow \alpha \rightarrow \pi} b^{\# \alpha} m^{\#(\pi\alpha)-1} \notag\\
&=& \sum_{\alpha\in\NC(p)} b^{\# \alpha} m^{e(\alpha)}. \label{moment-BN}
\end{eqnarray} 
Here, for the second equality we used $1+e(\alpha)= \#(\pi\alpha)$, 
which was proven in \cite[Lemma 2.1]{BanicaNechita2012}, where
$e(\cdot)$ is the number of blocks whose size is even. The latter formula gives the free cumulants of the limiting measure as 
\begin{equation}
\label{eq:reference}
k_r= \begin{cases} 
bm &\text{if $r$ is even,} \\ 
b &\text{if $r$ is odd.}  
\end{cases} 
\end{equation}

If we set
\[
x = \frac{b(m+1)}{2}; \qquad y = \frac{b(m-1)}{2}
\]
and define $X$ (resp.~$Y$) to be a random variable with free Poisson distribution
$\nu_{x,1}$ (resp.~$\nu_{y,1}$) then, if $X$ and $Y$ are free then
the cumulants of $X-Y$ are given by
\[
k^{(X-Y)}_r = x + (-1)^r y 
= 
\begin{cases} 
bm &\text{if $r$ is even,} \\ 
b &\text{if $r$ is odd.}  
\end{cases} 
\]
so they coincide with the cumulants \eqref{eq:reference} of the limiting distribution. In this way we showed convergence in expected moments.

Almost sure convergence can be proven similarly as in the proof of Theorem \ref{limit-dist}. As the limiting distribution is compactly supported, the convergence in moments implies weak converge.
\end{proof}

\section{Concluding remark} \label{conclusion} 
In this paper, we investigated asymptotic behavior of eigenvalues of
partial transpose of random quantum states on  bipartite systems. 
We naturally picked up three regimes depending on how the concerned spaces grow and
showed their connections to the semicircular distribution, the free Poisson distribution and the meander polynomials. 
Other regimes may show other interesting behaviors.

\appendix

\section{Lemmas for Section \ref{conv-ex}}

\begin{lemma}\label{identity1}
For arbitrary integers $p,t\geq 0$ and $D\geq 2$
the following bound holds true:
\begin{equation}
\label{eq:estimate}
\sum_{k\geq 0}   \binom{p}{k}    \binom{k}{t} (-1)^{p-k}\  \f(D,k) 
\leq\binom{p}{t} \left( \frac{p}{\sqrt{D}} \right)^{p-t},
\end{equation}
where $\f$ is given by \eqref{function-f}.
\end{lemma} 

\begin{proof}
We denote by $S$ the \emph{shift operator} on functions of a single variable, i.e.~$S:g(\cdot) \mapsto g(\cdot+1)$ or alternatively $(Sg)(x)=g(x+1)$. Then,
writing $\f(k):=\f(D,k)$, the left-hand side of \eqref{eq:estimate} becomes
\begin{eqnarray}
\sum_{k \geq 0}   \binom{p}{k}    \binom{k}{t} (-1)^{p-k}  \f(k)
&=& \sum_{k\geq 0} \binom{p}{k}  \left. S^k  \binom{\cdot}{t} \f(\cdot)\ (-1)^{p-k}\right|_{\ \cdot\ = 0}  \notag \\
&=& \Delta^p \left.  \binom{\cdot}{t}\f(\cdot) \right|_{\ \cdot\ = 0}
\label{cal-1},
\end{eqnarray}
where $\Delta = S-1$ denotes the \emph{forward difference operator}. 

Firstly, we will recover the well known product rule for the finite difference $\Delta$. For arbitrary functions $g,h$ of a single variable we have
\begin{eqnarray*}
\Delta \left[ g h \right]
&=& (Sg)(Sh)- g h \\
&=& [(Sg)-g] (Sh) + g \left[ (Sh) -h \right] \\ 
&=& (\Delta g) (Sh) + g(\Delta h) \\
&=& M (\Delta g \otimes S h + g \otimes \Delta h ),
\end{eqnarray*}
where $M(g\otimes h) = g h$ denotes the pointwise multiplication of functions.
In other words, we showed that
\[
\Delta M = M (\Delta \otimes S + 1 \otimes \Delta).
\]
It follows that higher powers of the forward difference operator act on products as follows:
\[
\Delta^p (gh) = (\Delta^p M) (g\otimes h)  
= \sum_{l=0}^p \binom{p}{l} \times \Delta^l g \times S^l \Delta^{p-l} h. 
\]
We used the fact that the operators of shift and the forward difference commute: $\Delta S = S \Delta$.

By applying this general formula in our particular setup we obtain:
\begin{eqnarray}
\eqref{cal-1} & = & \Delta^p \left. \binom{\cdot}{t} \f(\cdot) \right|_{\cdot = 0} \notag \\
&=& \sum_{l=0}^p \binom{p}{l} \times 
\underbrace{\left. \Delta^l \binom{\cdot}{t}  \right|_{\ \cdot\ =0} }_{(\diamondsuit)}
\times \left(S^l \Delta^{p-l}\right) \f(\cdot) \Big|_{\ \cdot\ =0} \notag \\
&=& \binom{p}{t} \left(S^t \Delta^{p-t}\right) \f(\cdot) \Big|_{\ \cdot\ =0} \notag \\
&=&\binom{p}{t} (\Delta^{p-t}  \f)(t), 
\label{cal-2}
\end{eqnarray}
where we used the following property of $(\diamondsuit)$:
\[
(\diamondsuit) = {\left. \Delta^l \binom{\cdot}{t}  \right|_{\ \cdot\ =0} } = \binom{0}{t-l} = [l=t]. 
\]
Thus our problem is reduced to estimating the quantity $(\Delta^{p-t}  \f)(t)$ appearing on the right-hand side of \eqref{cal-2}.

If $g$ is an arbitrary function of a single variable then
$\Delta$ acts on the product $g \times (S^i \f)$ as follows:
\begin{eqnarray*}
\Delta \big[g \times (S^i \f)\big] (k) &=& g(k+1) \f(k+i+1) - g(k)\f(k+i)  \\
&=& \left [ g(k+1) - \left( 1+ \frac{k+i}{D} \right)g(k) \right] \f(k+i+1) \\
&=& \left[ \left(\Delta -\frac{k+i}{D} \right)g\right] (k) \times (S^{i+1}\f)(k), \\ 
\end{eqnarray*}
where it follows from the definition \eqref{function-f} that 
\[
\f(k+i) = \left(1+ \frac{ k+i}{D}\right) \f(k+i+1).
\]
Hence inductively we get
\[
[\Delta^q  (g \times \f)](k) 
= \left[ \left( \Delta - \frac{k+q-1}{D}\right) \cdots \left( \Delta - \frac{k}{D}\right) g \right]
 \times (S^{q} \f)(k). 
\]
We are interested in the special case of this formula for $g=1$ given by the constant function:
\[
[\Delta^q   \f](k) 
= \underbrace{\left[ \left( \Delta - \frac{k+q-1}{D}\right) \cdots \left( \Delta - \frac{k}{D}\right) 1
\right]}_{(\spadesuit)} \times \f(k+q). 
\]

We will now analyze the structure of expression ($\spadesuit$).
We use the shorthand notation
\[ P_{i}=P_{i}(k) = - \frac{k+i}{D}.\]
Expression ($\spadesuit$) is a product of $q$ factors, each factor being the sum of two terms.
Let us expand this product; each of $2^q$ resulting summands is a product of the forward difference operators $\Delta$ (let us say that there are $r$ factors of this form) and expressions of the form $(P_i)_{0\leq i \leq q-1}$ (let us say that there are $q-r$ factors of this form).
Notice that these two expressions do not commute so the order of the factors is important. 
In the following we will study in detail expressions of this form.

Clearly $\Delta P_{i}= - \frac{1}{D}$ thus
a straightforward application of the product rule shows that
\begin{multline*}
\Delta P_{i(1)} \cdots P_{i(\ell)} =  
 \sum_{1\leq r\leq \ell}  
\tikz[baseline]{\node[anchor=base](n1){$\Delta$};}
P_{i(1)} \cdots
\tikz[baseline]{\node[anchor=base](n2){$P_{i(r)}$};}
\cdots P_{i(\ell)+1}  = \\ 
- \frac{1}{D} 
 \sum_{1\leq r\leq \ell}  
P_{i(1)} \cdots P_{i(r-1)} P_{ i(r+1)+1} \cdots P_{i(\ell)+1 }. 
\tikz[overlay]{
\coordinate (m1) at ($(n1)+(0,20pt)$);
\coordinate (m2) at (n2|-m1);
\draw[thick,->] (n1) -- (m1) -- (m2) -- (n2);
}
\end{multline*}
The right-hand side can be interpreted as follows: we sum over all ways of matching the forward difference operator $\Delta$ with one of the factors $P_{i(r)}$ on its right;
this factor is removed and every term $P_j$ on the right of $P_{i(r)}$ should be replaced by $P_{j+1}$. This matching has been illustrated graphically as an arrow connecting the difference operator $\Delta$ with the factor $P_{i(r)}$ on which it acts.

This observation can be extended to general products which we consider. Namely, we sum over all ways of matching $r$ difference operators $\Delta$ with factors $(P_i)$ in such a way that each operator $\Delta$ is matched with some factor $P_i$ which is on its right and each factor $P_i$ is used at most once. The factors $(P_i)$ which are matched should be removed, each unmatched factor $P_j$ (there are $q-2r$ of them) should be replaced by $P_{j+\delta}$ where $\delta$ denotes the number of factors $(P_i)$ which are matched and are on the left from $P_j$, thus $j+\delta\leq q-1$. Also, there is additional factor $\left( - \frac{1}{D}\right)^r$.

We can illustrate this by an example for $r=2$: one of the summands contributing to the product $P_2\ \Delta\ P_5\ \Delta\ P_8\ P_{13}\ P_{17}\ P_{21}\ P_{25}$ is given by:
\begin{multline*}
P_2\ \Delta\ P_5\ \Delta\ P_8\ P_{13}\ P_{17}\ P_{21}\ P_{25} = \\ \cdots+ 
 P_{2} 
\tikz[baseline]{\node[anchor=base](a1){$\Delta$};}
P_{5}
\tikz[baseline]{\node[anchor=base](b1){$\Delta$};}
P_{8}
\tikz[baseline]{\node[anchor=base](p1){$P_{13}$};}
P_{17}
\tikz[baseline]{\node[anchor=base](q1){$P_{21}$};}
P_{25}+
\cdots
\tikz[overlay]{
\coordinate (a2) at ($(a1)+(0,20pt)$);
\coordinate (p2) at (p1|-a2);
\draw[thick,->] (a1) -- (a2) -- (p2) -- (p1);
\coordinate (b2) at ($(b1)+(0,30pt)$);
\coordinate (q2) at (q1|-b2);
\draw[thick,->] (b1) -- (b2) -- (q2) -- (q1);
}
\tikz{\draw[white] rectangle (0,40pt);}= \\
\cdots+\frac{1}{D^2} P_2\ P_5\ P_8\ P_{18}\ P_{27}+ \cdots
\end{multline*}

The above two-step procedure (selecting one of $2^q$ summands with $r$ factors $\Delta$, then summing over the matching $\mathcal{P}$) can be clearly replaced by summing simply over $\mathcal{P}$ which should be a
partition of the set $[q]$ with $r$ blocks of length $2$ and $q-2r$ blocks of length $1$ (the places where the difference operator $\Delta$ occur are exactly the left elements of two-element blocks of $\mathcal{P}$). The number of such partitions $\mathcal{P}$ is equal to
\[ \frac{q (q-1) \cdots (q-2r+1) }{2^r r!} \leq \frac{q^{2r}}{2^r r!}.\] 
It follows that for $k\geq 0$
\[ \big\lvert (\spadesuit) \big\rvert 
\leq \sum_{0\leq r\leq \lfloor \frac{q}{2} \rfloor}   \frac{1}{D^r} \frac{q^{2r}}{2^r r!} \left( \frac{k+q}{D} \right)^{q-2r}.
\]

Thus for $q=p-t$ and $k=t$,
\[
\eqref{cal-2} = \binom{p}{t} (\Delta^{q}  \f)(t)=  \binom{p}{t}\ (\spadesuit) \ \F(p) 
\]
and 
\[
\big| \eqref{cal-2} \big|\leq \binom{p}{t} \sum_{0\leq r\leq \lfloor \frac{q}{2} \rfloor}  \frac{1}{D^r} \frac{p^{2r}}{2^r r!} \left( \frac{p}{D} \right)^{q-2r}=
\binom{p}{t} \sum_{0\leq r\leq \lfloor \frac{q}{2} \rfloor}  
\frac{p^q}{D^{q-r}}. 
\]
The sum on the right-hand side is dominated by its last summand multiplied by $2$, therefore 
\[
\big| \eqref{cal-2} \big|\leq 2 \binom{p}{t} \frac{p^q}{D^{q/2}}.
\]
which finishes the proof.
\end{proof} 

\begin{lemma}\label{cardinarity}
Let $h=h_p$ be a function defined on $S_p$ as in \eqref{function-ht}
and
\[
T_{p,h} = \{ \alpha \in \tilde S_p : h(\alpha) =h \}. 
\]
Then
\[
|T_{p,h}| \leq  4^{ \frac{p}{2} -1} p^{12h + 5}. 
\]
\end{lemma} 
\begin{proof}
We assume for a while that $p$ is an even number. 
We define
\[
k(\alpha) = 2|\alpha| -p.
\]
It fulfills the bound
\begin{equation}
\label{bound-a}
2h(\alpha) \geq 2|\alpha| + |\pi^2| - 2p +2 = 2|\alpha| -p = k(\alpha).
\end{equation}

Let $\sigma$ be an arbitrary permutation and $\tau$ be a transposition.
If $|\tau \sigma| < |\sigma|$, we say that $\sigma':=\tau \sigma$ was obtained by a \emph{cut}. This corresponds to splitting one of the cycles of $\sigma$ into two cycles; we say that \emph{the cycle was cut}. 
If $|\tau \sigma| > |\sigma|$, we say that $\sigma':=\tau \sigma$ was obtained by a \emph{gluing}. This corresponds to merging two of the cycles of $\sigma$ into one cycle.

Let $\alpha\in T_{p,h}$ be given. If $\alpha$ has a cycle of length at least $3$, we cut it; we iterate the procedure until we obtain permutation $\alpha'$ which consists only of cycles of length $1$ and $2$.
The number of cuts we need is upper-bounded by $k(\alpha)$.
To see this,
\[
(\text{the number of cuts for $\alpha$}) \leq \sum_{c \in \alpha} (|c| -2 ) = p-2\#(\alpha) = k(\alpha).
\]
Here, $|c|$ is the cardinality of a cycle $c$ of the permutation $\alpha$.
Importantly, $h(\alpha') \leq h(\alpha)$. 

Next, we glue fixed points of $\alpha'$ pairwise to get $\alpha''$ which consists of $p$ disjoint transpositions;
in the language of partition this is a (possibly crossing) pairing. 
The number of glues is upper-bounded by $k(\alpha)/2$;
\[
(\text{the number of glues}) 
\leq \frac{1}{2}\sum_{c \in \alpha} (|c| -2 ) = \frac{1}{2} k(\alpha). 
\]
Above, we used the fact that $\alpha$ has no fixed points. 

Write $g(\alpha) = g_p^{(2)}$, which is the notation used in Theorem \ref{HZ}.
Each operation of gluing can increase the genus at most by $1$; 
furthermore $g \leq h$ so that
\[
g (\alpha'') \leq  g (\alpha') + \frac{1}{2}k(\alpha) \leq 2 h(\alpha).
\]

To summarize:
we constructed a map $\alpha\mapsto\alpha''$ where $\alpha\in T_{p,h}$, with a property that
$g(\alpha'') \leq 2h$ and $\alpha''$ is an involution without fixed points.
Furthermore, $\alpha$ can be obtained from $\alpha''$ by multiplying by
at most $3h$ transpositions; this means that the preimage of any $\alpha''$ consists of at most $\binom{p}{2}^{3h}$ elements.
Then, by using Lemma \ref{HZcor}, 
\[
|T_{ph}| \leq \binom{p}{2} ^{3h} \times \sum_{g =0}^{2h}4^{\frac{p}{2}-1} p^{3g}
\leq 4^{\frac{p}{2}-1} p^{12h + 3} 
\]
which finishes the proof in the case when $p$ is even.

When $p$ is odd number,
we select some cycle consisting of more than two elements and
remove one of these elements. 
Then, we can do the same surgeries on $\lfloor \frac{p}{2} \rfloor$ points as before. 
This time, to recover $\alpha$ from $\alpha''$ we need additional step
so that we increase the power of $p$ in the bound by $2$. 
\end{proof}

Finally we prove Lemma \ref{HZcor} by using:
\begin{thm}[Harer-Zagier formula \cite{HarerZagier1986}]\label{HZ}
Let $\alpha \in S_{2n}$ and define the genus $g$ of $\alpha$ such that
\[ 
2g (\alpha) = |\alpha| +|\alpha^{-1} \pi| - 2n+1, 
\]
where $\pi = (1,2,\ldots,2n-1,2n)$.
Then, the number of involutions without fixed points with genus $g$, 
denoted by $\epsilon_g(n)$,  has 
the following recursive formula:  
\[
(n+1)\epsilon_g(n) = 2(2n-1) \epsilon_g(n-1) + (2n-1)(n-1) (2n-3) \epsilon_{g-1}(n-2). 
\] 
Here, the boundary condition for the recurrence is given by $\epsilon_0(n) = \Cat_n$, the Catalan numbers.
\end{thm}

\begin{lemma}\label{HZcor}
We have
\[
\epsilon_g (n) \leq 4^{n-1} n^{3g}
\]
for $g \geq 0$. 
\end{lemma} 
\begin{proof}
First, the bound is true for $g=0$. This comes from the fact that
\[
\epsilon_0(1) = \Cat_1 =1 \quad\text{ and }\quad
(n+1) \Cat_n = 2(2n-1)\Cat_{n-1}.
\]  
Next, we assume that the bound is true for $g-1$. Theorem \ref{HZ} implies 
\begin{eqnarray*}
\epsilon_g (n) 
&\leq& 4\epsilon_g(n-1) + 4(n-1)^2 \epsilon_{g-1}(n-2)  \\
&\leq& 4 \left[  4\epsilon_g(n-2) + 4(n-2)^2 \epsilon_{g-1}(n-3) \right] + 4(n-1)^2 \epsilon_{g-1}(n-2) \\
&\leq& \sum_{j=0}^{n-1} 4^{n-j} j^2 \epsilon_{g-1} (j-1). 
\end{eqnarray*}
By the induction hypothesis,
\begin{eqnarray*}
\epsilon_g (n) 
&\leq& \sum_{j=0}^{n-1} 4^{n-j} j^2  4^{j-2} (j-1)^{3g-3} \\
&\leq& 4^{n-2} \sum_{j=0}^{n-1} j^{3g-1} 
\leq 4^{n-2} \int_{0}^n x^{3g-1} dx
\leq 4^{n-1} n^{3g}.
\end{eqnarray*}
\end{proof}

\section{Measure concentration} \label{measure-concentrate} 
In this section we collect some results from asymptotic geometric analysis and random matrix theory,
which we use in this paper. 

The following theorem states that 
continuous functions on high-dimensional unit spheres 
have almost constant value except for sets of small measure: 
\begin{lemma}[Levy's lemma \cite{Levy1951}]\label{Levy-lemma}
Let $f: \mathbb S^k \rightarrow \mathbb R $ be a function with Lipschitz constant $L$, 
then there exists a constant $0<c_0<1$ and
\[
\Pr \{x \in \mathbb S^k : |f(x) - {\rm Median}[f]| \geq \epsilon  \} 
< \exp \left\{ \frac{ - c_0 (k-1)\epsilon^2}{L^2} \right\}.  
\] 
\end{lemma}  

It is well-known that singular values of random matrices of Gaussian entries are 
asymptotically concentrated \cite{MarcenkoPastur1967}. 
For example, a sharp estimate is found in the proof of Lemma 7.3 in \cite{HT03}
and we restate it in our setting: 
\begin{lemma}[Haagerup, Thorbj{\o}rnsen]\label{HT-lemma}
Let $\frac{mn}{l} = a_{l} $ and for $\epsilon > 0$ we have
\[
\Pr \left\{\lambda_{\max} (\tilde\rho_{l}) 
\geq \frac{1}{l} \left[ \left(\sqrt{a_l^{-1}} +1\right)^2 +\epsilon \right] \right\}
\leq a_l l \exp \left\{ -\frac{a_ll\epsilon^2}{4(a_l^{-1} +1)} \right\}. 
\]
Here, $\tilde \rho_l = \frac{W_l}{lmn}$ and 
$\lambda_{\max}(\tilde \rho)$ is the maximum eigenvalue of $\tilde \rho$.  
\end{lemma}
A similar phenomenon is also true even when they are normalized to have the Hilbert-Schmidt norm one, 
and we have the following corollary:
\begin{corollary} \label{cor-HT}
Let $\frac{mn}{l} = a_{l} $, which converges. Then, there exist two constants $c,C>0$ such that 
\[
\Pr \left\{\lambda_{\max} (\rho_{l}) \geq \frac{C}{l} \right\}
\leq \exp \left\{ -cl \right\}. 
\]
for large enough $l\in\N$. Here, $\rho_l = \frac{W_l}{\trace W_l}$.
\end{corollary}

\section*{Acknowledgments}

We thank the anonymous referee 
who not only suggested use Wick formula (instead of Weingarten calculus which was used in the preliminary version of the current paper) 
but also gave other useful comments. 

The research of M.F. was financially supported by the CHIST-ERA/BMBF project CQC. 

In the initial phase of research, P.\'S.~was a holder of a fellowship of \emph{Alexander von Humboldt-Stiftung}.
P.\'S.'s research has been supported by a grant of \emph{Deutsche Forschungsgemeinschaft} (SN 101/1-1).


\end{document}